\documentclass[11pt]{article}
\usepackage{amsmath,amsthm,amssymb}
\usepackage{mathrsfs}
\usepackage{array, longtable}
\usepackage{pdfsync}
\newtheorem{Theorem}{Theorem}[section]
\newtheorem{lem}[Theorem]{Lemma}
\newtheorem{Remark}[Theorem]{Remark}
\newtheorem{Definition}[Theorem]{Definition}
\newtheorem{Corollary}[Theorem]{Corollary}

\newtheorem{Example}[Theorem]{Example}

\numberwithin{equation}{section}
\newcommand{\dst}[2]{\genfrac{[}{]}{0pt}{0}{#1}{#2}}
\begin{document}

\title {Optimal Few-GHW Linear Codes and Their Subcode Support Weight Distributions 
}

\author{Xu Pan$^1$,~Hao Chen$^2$,~Hongwei Liu$^3$,~Shengwei Liu$^4$ 
}
\date{\small
$^1$School of Information Science and Technology/Cyber Security,
Jinan University, Guangzhou, Guangdong, 510632, China, panxucode@163.com\\
$^2$School of Information Science and Technology/Cyber Security,
Jinan University, Guangzhou, Guangdong, 510632, China, haochen@jnu.edu.cn\\
$^3$School of Mathematics and Statistics, Key Lab NAA-MOE, Central China Normal University, Wuhan, Hubei, 430079, China, hwliu@ccnu.edu.cn\\
$^4$School of Mathematics and Statistics, Central China Normal University, Wuhan, Hubei, 430079, China, shengweiliu@mails.ccnu.edu.cn\\
}

\maketitle

\begin{abstract}
Few-weight codes have been constructed and studied for many years, since their fascinating relations to finite geometries, strongly regular graphs and Boolean functions. Simplex codes are one-weight Griesmer $[\frac{q^k-1}{q-1},k ,q^{k-1}]_q$-linear codes and  they meet all Griesmer bounds of the generalized Hamming weights of linear codes. All the subcodes with dimension $r$ of a $[\frac{q^k-1}{q-1},k ,q^{k-1}]_q$-simplex code have the same subcode support weight $\frac{q^{k-r}(q^r-1)}{q-1}$ for $1\leq r\leq k$. In this paper, we construct linear codes meeting the  Griesmer bound of the $r$-generalized Hamming weight, such codes do not meet the Griesmer bound of the $j$-generalized Hamming weight for $1\leq j<r$. Moreover these codes have only few subcode support weights. The weight distribution and the subcode support weight distributions of these distance-optimal codes are determined. 
Linear codes constructed in this paper are natural generalizations of distance-optimal few-weight codes.


\end{abstract}

\section{Introduction}
Throughout this paper, let $\mathbb{F}_q$ be the finite field of order $q$, where $q=p^e$ and $p$ is a prime. For a positive integer $n$, let $\mathbb{F}_q^{n}$ be the $n$-dimensional vector space over $\mathbb{F}_q$ consisting of all the $n$-tuples $(x_{1},x_{2},\cdots,x_{n}),$ where $x_i\in \mathbb{F}_q$ for all $1\leq i\leq n.$ For a real number $a$, the smallest integer that is greater than or equal to $a$ is denoted by $\lceil a\rceil$.

\subsection{The generalized Hamming weights of linear codes}

A nonempty subset $X$ of $\mathbb{F}_q^{n}$ is called an $(n,|X|)_q$-{\it code}, where $|X|$ is the size of $X$. For a subcode $X$, the {\it subcode support} of $X$ is $$\mathrm{supp}(X)=\{1\leq  i\leq n \,|\,\exists \,\mathbf{x}=(x_{1}, x_2, \cdots,x_{n})\in X\text{ with } x_{i}\neq0\},$$
and the {\it subcode support weight} of $X$ is $w(X)=  |\mathrm{supp}(X)|.$
For any vector $\mathbf{x}\in \mathbb{F}_q^{n}$,   the subcode support weight of the subspace generated by $\mathbf{x}$ is equal to the Hamming weight of $\mathbf{x}$ denoted by $w(\mathbf{x}).$

A subspace $C$ of $\mathbb{F}_q^{n}$ with dimension $k$ is called an {\it $[n,k]_q$-linear code} or an {\it $[n,k,d]_q$-linear code}, where $d$ is the minimum Hamming distance of $C$.
An $[n, k,d]_q$-linear code $C$ is called {\it distance-optimal} if no $[n, k,\tilde{d}]_q$-linear code such that $\tilde{d}> d$ exists.
For an $[n,k]_q$-linear code $C$ and $1\leq r\leq k$, the  {\it$r$-generalized Hamming weight} ($r$-GHW) of $C$ is defined as $$d_{r}(C)=\min\{|\mathrm{supp}(U)|\,\big|\,\text{$U$ is a subspace of $C$ with dimension $r$}\}.$$
The set $\{d_{1}(C),d_{2}(C),\cdots,d_{k}(C)\}$ is called the  {\it weight hierarchy} of $C$.
When $r=1$, the parameter $d_{1}(C)$ is the minimum Hamming distance of $C$.

In 1977, the notion of the generalized Hamming weights was introduced  by Helleseth, Kl{\o}ve and Mykkeltveit in \cite{HT}. Wei \cite{W} provided an application of the generalized Hamming weights in wire-tap channels of type II.
Since then, lots of works have been done in computing and describing the generalized Hamming weights for many classes of linear codes, such that Hamming codes \cite{W},
Reed-Muller codes \cite{B,HP1,W},
 binary Kasami codes \cite{HK95},
 Melas and dual Melas codes \cite{V1},
 cyclic codes \cite{CC,DS,FT,L,MP,SC95,V,V2,V3,YF},
 some trace codes \cite{CM,SV}, 
 some algebraic geometry codes \cite{V4,V5,YK,TV}, 
  linear codes relating to the defining set and  quadratic forms \cite{JFW,L1,L2,LZ} 
 and linear codes defined by simplicial complexes \cite{LZW}.
Liu and Pan introduced the notion of the generalized pair weights and the notion of the generalized $b$-symbol weights of linear codes, which are generalizations of the minimum symbol-pair weight and the minimum $b$-symbol weight of linear codes in \cite{LP,LP1}. In \cite{Luo},  Luo {\it et al.} introduced the relative generalized Hamming weights and gave some new characterizations on the wire-tap channel of type II. 
Liu and Wei gave some necessary conditions and the explicit construction for a
linear code to be optimal on the relative generalized Hamming weight \cite{LW}.

\subsection{The Griesmer bound for the generalized Hamming weights}
For an $[n,k]_q$-linear code $C$,
the {\it Griesmer bound for the generalized Hamming weights} \cite{HK,HK92,HK95} is provided as following:
$$d_r(C)+\sum_{i=1}^{k-r}\lceil \frac{(q-1)d_r(C)}{q^i(q^r-1)}\rceil \leq  n,$$
where $d_r(C)$ is $r$-GHW of $C$ for $1\leq r\leq k$.
When $r=1$, this bound is the Griesmer bound for the minimum Hamming weight.
The Griesmer bound for the relative generalized Hamming weights  is proved in \cite{Z,Z1}.
The parameter $$\delta_r(C)=n-d_r(C)-\sum_{i=1}^{k-r}\lceil \frac{(q-1)d_r(C)}{q^i(q^r-1)}\rceil$$ 
is called the {\it $r$-Griesmer defect} of $C.$ If $r=\min\{1\leq j\leq k\,|\,\delta_j(C)=0\}$ for an $[n,k]_q$-linear code $C$, then $C$ is called a {\it $r$-Griesmer code}.
If  an $[n,k]_q$-linear code $C$ satisfies $|\mathrm{supp}(C)|=n$, then there always exists an integer $r$ such that $C$ is a $r$-Griesmer code.
A $r$-Griesmer code $C$ is optimal, since the $s$-GHW $d_s(C)$ of $C$ reaches the maximal value for $s\ge r$.
For example, the $[24,12,8]_2$-Golay code $C$ with the  weight hierarchy 
$$ \{8,12,14,15, 16, 18, 19, 20, 21, 22, 23, 24]\}$$ given in \cite{TV} satisfies 
$\delta_r(C)=1 \text{ for } 1\leq r\leq 5 \text{ and }\delta_r(C)=0 \text{ for } 6\leq r\leq 12.$
And there is  a $[51,8,24]_2$-cyclic code $C$ with the  weight hierarchy $$\{  24,36,42,45,47,48,50,51 \}$$ in Example 4.5 of \cite{V} such that $\delta_r(C)=1$ for  $ 1\leq r\leq 6 $ and $ \delta_7(C)=\delta_8(C)=0. $

The $1$-Griesmer codes are also known as {\it Griesmer codes}. And Griesmer codes are always distance-optimal.
 Griesmer codes have been studied for many years due to not only their optimality but also
 their geometric applications \cite{D151,D18}. 
We refer the readers to \cite{HLL,HLZ,HXL,LYT,LDT,WZD,WLX,WZY} for works on the construction of Griesmer codes.
An $[n,k]_q$-linear code $C$ with $\delta_1(C)=1$ is also called an {\it almost Griesmer code}.


\subsection{The subcode support weight distributions of linear codes}
Let $C$ be an $[n,k]_q$-linear code,
the sequence $[A_1^{r}(C), A_2^{r}(C),\cdots, A_n^{r}(C)]$ is called the {\it $r$-subcode support weight distribution} ($r$-SSWD) of $C$,
where $$A_j^{r}(C)=|\{ U\,|\,U\text{ is a subspace of $C$ with dimension $r$ and $w(U)=j$} \}|,$$ for $1\leq j\leq n$ and $1\leq r\leq k.$
 Helleseth, Kl{\o}ve and Mykkeltveit \cite{HT} proved an interesting fact that the weight distribution of
the lifted code of a linear code is related to the subcode support weight distributions of this linear code. And this result was also discovered independently by Greene \cite{GC}.

The subcode support weight distributions of a linear code is a generalization of the
weight distribution of a linear code. As the weight distribution of a linear code, the subcode support weight distributions of a linear code are closely related to that of its dual code, and such a relationship was given in \cite{K1,OT}.
In \cite{SL}, Shi, Li and Helleseth determined  the subcode support weight distributions of the projective Solomon-Stiffler codes by using some combinatorial properties of subspaces.
Luo and Liu determined the subcode support weight distributions of some optimal linear  codes  constructed by down-sets in \cite{LL}.

\subsection{Few-weight codes}

The {\it weight distribution} of an $(n,|C|)_q$-code $C$ is defined as
the sequence $$[A_0(C),A_1(C), A_2(C),\cdots, A_n(C)],$$
where $A_j(C)=|\{ \mathbf{c}\in C\,|\, w(\mathbf{c})=j\}|$ for $0\leq j\leq n$.
Note that, if $C$ is a linear code, then $A_j(C)=(q-1)A_j^1(C)$  for $1\leq j\leq n$.
The weight distribution not only contains important information about
the error correcting capability of the code, but also allows
the computation of the error probability of error detection
and correction of this code \cite{K}. 
An $(n,|C|)_q$-code $C$ is said to
be a {\it $t$-weight code} if  $$t=|\{ j  \,|\,A_j(C)\neq 0 \text{ and } 1\leq j\leq n\}|.$$
 Few-weight  codes have applications in secret sharing schemes \cite{A,CA} and authentication codes \cite{D07,DW}.
In \cite{CA1}, Calderbank and Goethals studied $3$-weight codes and association schemes.
The $2$-weight codes have close connection with certain strongly regular graphs \cite{CA2}.
There have been many papers on constructing optimal $t$-weight linear codes, from Boolean functions, difference sets, simplicial complexes, and coding over finite rings, see \cite{D151,D15,HLL,HXL,LM19,M1,M2,MQ}.
The weight distributions of these optimal $t$-weight codes are determined.
And there are some works relating to  few-weight  codes, see \cite{DK,DN,LS,MS,SX,TL,XF}.

\subsection{Our contributions and organization of this paper}
In this paper, we construct several families of distance-optimal few-weight $r$-Griesmer $[n,k]_q$-linear codes for any $1\leq r\leq k$, and their 
subcode support weight distributions are determined in \textbf{Theorem~\ref{xxd2}, Theorem~\ref{xxd}, Theorem~\ref{xxd4} and Theorem~\ref{chen}}. 
The method of our construction is to use
modified Solomon-Stiffler codes and simplex complement codes of GRS codes.
When we determine the subcode support weight distributions, the technique we used is to enumerate the number of subspaces satisfying the certain conditions.

The rest of the paper is organized as follows: Section 2 gives some preliminaries and some notations. In Section 3,  we provide a lemma which can be used to enumerate the number of subspaces satisfying certain conditions.
By this lemma, we determine the subcode support weight distributions of a family of Griesmer codes and a family of distance-optimal $r$-Griesmer codes for $r\ge 2$.
In Section 4,  we construct several infinite families of distance-optimal $r$-Griesmer codes and determine their subcode support weight distributions by another lemma. And some examples are provided.
In Section 5, we use simplex complement codes of GRS codes to construct a family of distance-optimal $2$-Griesmer codes, and their weight distributions and generalized Hamming weights are determined. 
Section 6 concludes this paper and lists some future works.
All computations in this paper were done in MAGMA V2.12-16.

\section{Preliminaries}
This section introduces some notions and basic properties used in this paper.
For an $[n,k]_q$-linear code $C$, the dual code $C^{\perp}$ of $C$ is defined as
$$
C^{\perp}=\{{\bf x}\in \mathbb{F}_q^n \,| \,  {\bf c}\cdot{\bf x} =0, \forall \, {\bf c}\in C\},
$$
where $ ``\cdot" $ is the standard Euclidean inner product.
For nonzero ${\bf y}\in \mathbb{F}_q^n$, the dimension of $\langle {\bf y}\rangle^{\perp}$ is $k-1,$
where $\langle {\bf y}\rangle$ is the linear subspace generated by ${\bf y}.$

When we determine the subcode support weight distributions of linear codes constructed in this paper, we use $q$-ary Gaussian binomial coefficients.

\begin{Definition}\label{KK1w}
For integers $0\leq r\leq k$, the $q$-ary Gaussian binomial coefficients $\dst{k}{r}_q$
are defined by: $$\dst{k}{r}_q=\frac{(q^k-1)(q^k-q)\cdots(q^k-q^{r-1})}{(q^r-1)(q^r-q)\cdots(q^r-q^{r-1})} .$$
\end{Definition}
If $r=0$ or $r=k$, the value $\dst{k}{r}_q$ is one. When $r<0$ or $r>k$, the value $\dst{k}{r}_q$ is defined as zero.
For a vector space $U$ over $\mathbb{F}_q$ with dimension $k$, assume ${\rm SUB}(U)=\{V\,|\, V\text{ is a nonzero subspace of }  U \}$ and
 $${\rm SUB}^{r}(U)=\{V\,|\,V\text{ is a subspace of $U$ with dimension }r\}$$ for $1\leq r\leq k.$
 Then $|{\rm SUB}^{r}(U)|=\dst{k}{r}_q  \,\,\text{
and }\,\,{\rm SUB}(U)=  \bigcup_{r=1}^{k} {\rm SUB}^{r}(U).$

 Let $G =[G_{1},G_{2},\cdots,G_{n}]$ be a $k\times n$ matrix,  where $G_i^T\in\mathbb{F}_q^{k}$ for all $1\leq i\leq n$. For any $V\in {\rm SUB}(\mathbb{F}_q^{k})$, the function $m_{G}: {\rm SUB}(\mathbb{F}_q^{k}) \to \mathbb{N}$ is defined as following:
$$
m_{G}(V)=|\{1\leq i\leq n\,|\, G^T_{i} \in V\}|.
$$

\begin{Remark}
Let $G$ and $\tilde{G} $ be matrices of the same column size and the columns of $G$ are contained in those columns of $\tilde{G}$. We denote the matrix obtained by puncturing the columns of $G$ from $\tilde{G}$ by $\tilde{G} \backslash G$. 
And we know that $m_{\tilde{G}}(V)=m_{G}(V)+m_{\tilde{G} \backslash G }(V)$ for any $V\in {\rm SUB}(\mathbb{F}_q^{k}) .$
\end{Remark}

Let $C$ be an $[n,k]_q$-linear code with a generator matrix $G=[G_{1},G_{2},\cdots,G_{n}]$. Then we know that, for any subspace $U\in {\rm SUB}^{r}(C)$, there exists a subspace $V\in {\rm SUB}^{r}(\mathbb{F}_q^k)$ such that $$U=\{ \mathbf{y}G\,|\,  \mathbf{y}\in V\}.$$
The following lemma is easy to be obtained.
\begin{lem} \label{weight}
Assume the notation is as given above. 
For an $[n,k]_q$-linear code $C$ and a subspace $U\in {\rm SUB}^{r}(C)$,
 the subcode support weight of $U$ is $$w(U)= n-m_{G}(V ^{\bot}).$$ 
And the $r$-GHW of $C$ is $d_r(C)=n-\max\{m_G(V)\,|\, V\in {\rm SUB}^{k-r}(\mathbb{F}_q^{k})\}$
for $1\leq r\leq k.$
\end{lem}


\section{The subcode support weight distributions of $r$-Griesmer codes}
In this section, we determine the subcode support
weight distributions of Griesmer codes ($1$-Griesmer codes) constructed in the first subsection,  
and we determine the subcode support
weight distributions of $r$-Griesmer codes for $r\ge 2$ constructed in the second subsection. 
 Recall 
 $${\rm SUB}^{r}(U)=\{V\,|\,V\text{ is a subspace of $U$ with dimension }r\}$$ for $1\leq r\leq k,$
where $\dim(U)=k$.
Assume integers $u_1,u_2, \cdots, u_s$ and $v_1,v_2, \cdots, v_s$ satisfy
$$0=u_0< u_1< u_2< \cdots< u_s\leq k\,\, \text{ and }\,\, 0=v_0\leq v_1\leq v_2\leq \cdots\leq v_s\leq k$$ 
with $v_i\leq u_i$ for $1\leq i\leq s.$ Let $U_i\in  {\rm SUB}^{u_i}(\mathbb{F}_q^{k})$ such that $$U_1\subseteq U_2\subseteq \cdots\subseteq U_s.$$
Let $N^{u_1, u_2, \cdots, u_s}_{v_1,v_2,\cdots, v_s}$ be the number of 
subspaces $V\in  {\rm SUB}^{v_s}(U_s)$ such that $\dim(U_i \bigcap V)=v_i$
for every $1\leq i\leq s,$ i.e., 
$$N^{u_1, u_2, \cdots, u_s}_{v_1,v_2,\cdots, v_s}
=\{V\in  {\rm SUB}^{v_s}(U_s)\,|\, \dim(U_i \bigcap V)=v_i \text{ for } 1\leq i\leq s\}.$$
 If $s=1,$
then $N^{u_1}_{v_1}=\dst{u_1}{v_1}_q$ by the definition.

In the following lemma, the integer $N^{u_1, u_2, \cdots, u_s}_{v_1,v_2,\cdots, v_s}$ is determined. 
\begin{lem} \label{num}
\begin{description}
  \item[(a)] Fix a subspace $V_{s-1}\in  {\rm SUB}^{v_{s-1}}(U_{s-1}),$
then $$|\{V_s\in  {\rm SUB}^{v_s}(U_s)\,|\,  U_{s-1} \bigcap V_s=V_{s-1}\}|=q^{(u_{s-1}-v_{s-1})(v_s-v_{s-1})}\dst{u_s-u_{s-1}}{v_s-v_{s-1}}_q.$$

\item[(b)]  For $s\ge 1,$ then $$N^{u_1, u_2, \cdots, u_s}_{v_1,v_2,\cdots, v_s}= \prod_{i=1}^sq^{(u_{i-1}-v_{i-1})(v_{i}-v_{i-1})}\dst{u_{i}-u_{i-1}}{v_{i}-v_{i-1}}_q .$$
  
\end{description}

\end{lem}

\begin{proof}
{\bf (a)} It is Lemma 2 of \cite{LL}.

{\bf (b)}  Assume $u_0=v_0=0.$
If $s=1$, then $$N^{u_1}_{v_1}=\dst{u_1}{v_1}_q=q^{(u_0-v_0)(v_1-v_0)}\dst{u_1-u_0}{v_1-u_0}_q.$$
Hence we assume $s\ge 2$ and work by induction on $s$.

Note that $U_1\subseteq U_2\subseteq \cdots\subseteq U_s$
and $U_{s-1}\bigcap V=U_{s-1}\bigcap( U_{s}\bigcap V).$
By Statement {\bf (a)} and the induction,  we have that 
\begin{eqnarray*}
N^{u_1, u_2, \cdots, u_s}_{v_1,v_2,\cdots, v_s}
&=&|\{V\in  {\rm SUB}^{v_s}(U_s)\,|\, \dim(U_i \bigcap V)=v_i \text{ for } 1\leq i\leq s\}|\\
&=&\sum_{V_{s-1}\in \{V\in  {\rm SUB}^{v_{s-1}}(U_{s-1})\,|\, \dim(U_i \bigcap V)=v_i \text{ for } 1\leq i\leq s-1\} }|\{V\in  {\rm SUB}^{v_s}(U_s)\,|\,  U_{s-1} \bigcap V=V_{s-1}\}|\\
&=&\sum_{V_{s-1}\in \{V\in  {\rm SUB}^{v_{s-1}}(U_{s-1})\,|\, \dim(U_i \bigcap V)=v_i \text{ for } 1\leq i\leq s-1\} }q^{(u_{s-1}-v_{s-1})(v_{s}-v_{s-1})}\dst{u_{s}-u_{s-1}}{v_{s}-v_{s-1}}_q\\
&=&q^{(u_{s-1}-v_{s-1})(v_{s}-v_{s-1})}\dst{u_{s}-u_{s-1}}{v_{s}-v_{s-1}}_qN^{u_1, u_2, \cdots, u_{s-1}}_{v_1,v_2,\cdots, v_{s-1}}\\
&=&q^{(u_{s-1}-v_{s-1})(v_{s}-v_{s-1})}\dst{u_{s}-u_{s-1}}{v_{s}-v_{s-1}}_q
\prod_{i=1}^{s-1}q^{(u_{i-1}-v_{i-1})(v_{i}-v_{i-1})}\dst{u_{i}-u_{i-1}}{v_{i}-v_{i-1}}_q\\
&=&\prod_{i=1}^sq^{(u_{i-1}-v_{i-1})(v_{i}-v_{i-1})}\dst{u_{i}-u_{i-1}}{v_{i}-v_{i-1}}_q .
\end{eqnarray*}
\end{proof}

\begin{Remark}
If $s=2,$ then $N^{u_1, u_2}_{v_1,v_2}=q^{(u_1-v_1)(v_2-v_1)}\dst{u_2-u_1}{v_2-v_1}_q\dst{u_1}{v_1}_q
,$ which is proved in Lemma 3.1 of \cite{SL}.
\end{Remark}

For $k \ge  2$, let $S_{q,k}$ be a $k \times  \frac{q^k-1}{q-1}$ matrix over $\mathbb{F}_q$ such that every two columns of $S_{q,k}$ are linearly independent.
This matrix generates the simplex code, which is a $[\frac{q^k-1}{q-1}, k, q^{k-1}]_q$-linear code.


\subsection{The subcode support weight distributions of Griesmer codes }

In the following theorem, we construct an infinite family of Griesmer codes, and we determine their subcode support weight distributions.
\begin{Theorem}\label{xxd2}
Assume the notation is as given above. For integers 
$$0=u_0< u_1< u_2< \cdots< u_s<u_{s+1}=k$$ and $1\leq s\leq t,$
 there exists  an $[n,k,d]_q$-linear code $C$ with  $$n=t\frac{q^{k}-1}{q-1}-\sum_{i=1}^s\frac{q^{u_i}-1}{q-1} \text{ and } d=tq^{k-1}-\sum_{i=1}^sq^{u_i-1}.$$
 
 For $1\leq r\leq k$, the linear code C satisfies the following properties:
 \begin{description}
   \item[(a)] The weight distribution of the linear code $C$ is completely determined in the following table:
$$\begin{array}{c|c}
                                         \text{weight} & \text{multiplicity} \\
                                           \hline
                                         0 & 1 \\
                                    d& q^{k} -q^{k-u_1}\\
                                          d+q^{u_1-1}& q^{k-u_1}-q^{k-u_2}\\
                                          d+q^{u_1-1}+q^{u_2-1}& q^{k-u_2}-q^{k-u_3}\\
                                          \vdots&\vdots\\
                                          d+\sum_{i=1}^{s-1}q^{u_i-1}&q^{k-u_{s-1}}-q^{k-u_s}\\
                                          d+\sum_{i=1}^sq^{u_i-1}&q^{k-u_s}-1\\
                                          \hline
                                        \end{array}  .$$

  \item[(b)]   The linear code $C$ is a Griesmer code.
 
\item[(c)] The $r$-GHW  of $C$ satisfies $$d_{r}(C)=t\frac{q^{k}-q^{k-r}}{q-1}-\sum_{i=1}^s\frac{q^{u_i}-1}{q-1}+
\sum_{i=j_r+1}^s \frac{q^{ u_i-r}-1}{q-1}  ,$$
where $j_r=\max\{ 0\leq i\leq s\,|\,u_i\leq r\}.$

\item[(d)]  Assume $v_0=0$ and $ v_{s+1}=k-r$. 
The $r$-SSWD of $C$ is
$$A_j^{r}(C) =\sum_{( v_1,v_2,\cdots, v_s,v_{s+1})\in \mathbf{Z}(q, m_{r,j}) }  \prod_{i=1}^{s+1}q^{(u_{i-1}-v_{i-1})(v_{i}-v_{i-1})}\dst{u_{i}-u_{i-1}}{v_{i}-v_{i-1}}_q.$$
where $m_{r,j}=j-t\frac{q^{k}-q^{k-r}}{q-1}+\sum_{i=1}^s\frac{q^{u_i}-1}{q-1}$ and 
 { \small $$ \mathbf{Z}(q, m_{r,j}) =\Biggm{\{ }( v_1,\cdots, v_s,k-r)\in \mathbb{Z}^{s+1}\,\Biggm{|}\, 
\left\{
  \begin{array}{ll}
    0\leq v_i-v_{i-1}\leq u_i-u_{i-1}, & \hbox{$1\leq i\leq s$;} \\
    \max\{0,\,u_i-r\}\leq v_i\leq \min\{u_i,\,k-r\}, & \hbox{$1\leq i\leq s$;} \\
    m_{r,j}=\sum_{i=1}^{s} \frac{q^{v_i}-1}{q-1}.& \hbox{}
  \end{array}
\right. \Biggm{\}}  .$$}
 \end{description}
 \end{Theorem}

\begin{proof}
Let $\mathbf{e}_i\in \mathbb{F}_q^k$ be the vector with all $0$s except for a $1$ in the $i$th coordinate.
Assume $U_i$ is the linear subspace generated by $\{\mathbf{e}_{1},\,\mathbf{e}_{2},\cdots ,\mathbf{e}_{u_i}\}$ for $1\leq i\leq s+1.$
Let $G_i$ be the $k \times \frac{q^{u_i}-1}{q-1}$ submatrix of $S_{q,k}$ such that each column of $G_i$ is in $U_i$.  
 Let $C$ be the $[n,k]_q$-linear code with the generator matrix 
\begin{equation}\label{de}
  G=[S_{q,k}\backslash G_1, S_{q,k}\backslash G_2,\cdots,S_{q,k}\backslash G_s,\underbrace{S_{q,k},\cdots,S_{q,k}}_{t-s}],
\end{equation}
where $n=t\frac{q^{k}-1}{q-1}-\sum_{i=1}^s\frac{q^{u_i}-1}{q-1}$.

For any subspace $U\in {\rm SUB}^{r}(C)$, there exists a subspace $V\in {\rm SUB}^{r}( \mathbb{F}_q^k)$ such that $U=\{ \mathbf{y}G\,|\,  \mathbf{y}\in V\}.$
By Lemma~\ref{weight}, we know that the subcode support weight of $U$ is 
\begin{equation}\label{uu31}
 w(U)= n-m_{G}( V^{\bot}),
\end{equation}
where $\dim(V^{\bot})=k-r.$
Assume $\tilde{G}$ is the $k\times t\frac{q^{k}-1}{q-1}$ matrix $[\underbrace{S_{q,k}, S_{q,k},\cdots,S_{q,k}}_t].$
By the definition of the function $m_{G}$, we have that
\begin{equation}\label{oor31}
  m_{\tilde{G}}(V^{\bot})=m_{G}(V^{\bot})+\sum_{i=1}^s m_{G_i}(V^{\bot})  .
\end{equation}

By the definition of $S_{q,k}$,
we know that $m_{\tilde{G}}(V^{\bot})=t\frac{q^{k-r}-1}{q-1} .$
By Equalities~(\ref{uu31}) and (\ref{oor31}), we have that
\begin{equation}\label{oo31}
  w(U)=t\frac{q^{k}-q^{k-r}}{q-1}-\sum_{i=1}^s\frac{q^{u_i}-1}{q-1}+\sum_{i=1}^s m_{G_i}(V^{\bot}).
\end{equation}

{\bf(a)}
Suppose $r=1.$
For any nonzero codeword ${\bf c}\in C$, there exists a unique nonzero vector ${\bf y}\in \mathbb{F}_q^k$ such that ${\bf c}={\bf y}G$.
By Equality~(\ref{oo31}), we know that the Hamming weight of $\mathbf{c}$ is 
\begin{equation}\label{oo1e}
  w(\mathbf{c})=tq^{k-1}-\sum_{i=1}^s \frac{q^{u_i}-1}{q-1}+\sum_{i=1}^s m_{G_i}(\langle \mathbf{y}\rangle^{\bot}).
\end{equation}

Since $\dim(U_i)=u_i$ and $\dim(U_i\bigcap \langle \mathbf{y}\rangle^{\bot})\in \{u_i-1,\,u_i\}$ for $1\leq i\leq s$,
we get that, for $1\leq i\leq s,$
 $$m_{G_i}(\langle \mathbf{y}\rangle^{\bot})\in \{\frac{q^{u_i-1}-1}{q-1},\frac{q^{u_i}-1}{q-1}\}.$$

Note that $$\{\mathbf{0}\}=U_0\subseteq U_1\subseteq U_2\subseteq \cdots\subseteq U_s\subseteq U_{s+1}=\mathbb{F}_q^k.$$

For any nonzero vector $\mathbf{y}\in \mathbb{F}_q^k$, there exists a unique integer $j_{\mathbf{y}}$ such that for $0\leq j_{\mathbf{y}}\leq s$ such that  $U_{ j_{\mathbf{y}}} \subseteq \langle \mathbf{y}\rangle^{\bot}$ and $U_{ j_{\mathbf{y}}+1} \nsubseteq\langle \mathbf{y}\rangle^{\bot}$.
Then $\sum_{i=1}^s m_{G_i}(\langle \mathbf{y}\rangle^{\bot})=\sum_{i=1}^{ j_{\mathbf{y}}} \frac{q^{u_i}-1}{q-1}+\sum_{i= j_{\mathbf{y}}+1}^s\frac{q^{u_i-1}-1}{q-1}.$
And we know that, for $0\leq j\leq s$, 
 \begin{eqnarray*}
&&   |\{{\bf y}\in \mathbb{F}_q^k\,|\,  \mathbf{y}\neq \mathbf{0},\,\, U_j \subseteq \langle \mathbf{y}\rangle^{\bot}\text{ and }U_{j+1} \nsubseteq \langle \mathbf{y}\rangle^{\bot}\}|\\
&=&   |\{{\bf y}\in \mathbb{F}_q^k\,|\,  \mathbf{y}\neq \mathbf{0}\text{ and } U_j \subseteq \langle \mathbf{y}\rangle^{\bot}\}|- |\{{\bf y}\in \mathbb{F}_q^k\,|\,  \mathbf{y}\neq \mathbf{0}\text{ and } U_{j+1} \subseteq \langle \mathbf{y}\rangle^{\bot}\}|\\
                          &=& q^{k-u_j}-q^{k-u_{j+1}}.
                        \end{eqnarray*}

By Equality~(\ref{oo1e}), we know that the minimum Hamming distance of $C$ is $$d=tq^{k-1}-\sum_{i=1}^s \frac{q^{u_i}-1}{q-1}+\sum_{i=1}^s\frac{q^{u_i-1}-1}{q-1}
=tq^{k-1}-\sum_{i=1}^s q^{u_i-1}.$$
And the weight distribution of $C$  is completely determined in the following table:
$$\begin{array}{c|c}
                                         \text{weight} & \text{multiplicity} \\
                                           \hline
                                         0 & 1 \\
                                    d& q^{k} -q^{k-u_1}\\
                                          d+q^{u_1-1}& q^{k-u_1}-q^{k-u_2}\\
                                    d+q^{u_1-1}+q^{u_2-1}& q^{k-u_2}-q^{k-u_3}\\
                                          \vdots&\vdots\\
                                          d+\sum_{i=1}^{s-1}q^{u_i-1}&q^{k-u_{s-1}}-q^{k-u_s}\\
                                          d+\sum_{i=1}^sq^{u_i-1}&q^{k-u_s}-1\\
                                          \hline
                                        \end{array}  .$$

{\bf(b)} Note that the minimum Hamming distance of $C$ is $d=tq^{k-1}-( q^{u_1-1}+q^{u_2-1}+\cdots +
q^{u_s-1}).$
Since 
\begin{eqnarray*}
  \lceil \frac{tq^{k-1}-( q^{u_1-1}+q^{u_2-1}+\cdots +
q^{u_s-1})}{q^j}\rceil &=&   \lceil tq^{k-1-i}-( q^{u_1-1-j}+q^{u_2-1-j}+\cdots +
q^{u_s-1-j})\rceil\\
   &=& tq^{k-1-i}-( q^{u_{s_j+1}-1-j}+q^{u_{s_j+2}-1-j}+\cdots +
q^{u_{s}-1-j}) ,
\end{eqnarray*}
where $u_0=0$ and  $s_j=\max\{ 0\leq i\leq s\,|\,u_i-1< j\},$
we have that  $$\sum_{j=1}^{k-1}\lceil \frac{d}{q^j}\rceil= t\frac{q^{k}-1}{q-1}-\sum_{i=1}^s\frac{q^{u_i}-1}{q-1}=n.$$
Hence $C$ is a Griesmer code.

{\bf(c)}  Assume $1\leq r\leq k$. By Equality~(\ref{oo31}), the $r$-GHW of $C$ is
\begin{eqnarray*}
   d_{r}(C)&=& \min\{w(U)\, |\, U\in {\rm SUB}^{r}(C)\} \\
   &=&  t\frac{q^{k}-q^{k-r}}{q-1}-\sum_{i=1}^s\frac{q^{u_i}-1}{q-1}+
\min\{\sum_{i=1}^s m_{G_i}(V^{\bot}) \,| V\in {\rm SUB}^{r}(\mathbb{F}_q^{k})\,\}\\
&=&  t\frac{q^{k}-q^{k-r}}{q-1}-\sum_{i=1}^s\frac{q^{u_i}-1}{q-1}+
\min\{\sum_{i=1}^s m_{G_i}(V) \,| V\in {\rm SUB}^{k-r}(\mathbb{F}_q^{k})\,\}.
\end{eqnarray*}

Then we determine $\min\{\sum_{i=1}^s m_{G_i}(V) \,| V\in {\rm SUB}^{k-r}(\mathbb{F}_q^{k})\,\}.$ 
Assume $\dim(V)=k-r$.
Since  $\dim(U_i)=u_i$ and $U_i+V\subseteq \mathbb{F}_q^k,$
we know that $$\max\{0,\,u_i-r\}\leq \dim(U_i\bigcap V)\leq \min\{u_i,\,k-r\}.$$


Let   $j_r=\max\{ 0\leq i\leq s\,|\,u_i\leq r\}.$
Since $m_{G_i}(V)=\frac{q^{ \dim(U_i\bigcap V)}-1}{q-1}$,
we have that $$\min\{\sum_{i=1}^s m_{G_i}(V) \,| V\in {\rm SUB}^{k-r}(\mathbb{F}_q^{k})\}
\ge \sum_{i=j_r+1}^s \frac{q^{ u_i-r}-1}{q-1}.$$

Let $V_{k-r}$ be the linear subspace generated by $\{\mathbf{e}_{r+1},\,\mathbf{e}_{r+2},\cdots,\mathbf{e}_{k}\}.$
Note that $$\dim(U_i\bigcap V_{k-r})= 0  $$ for $1\leq i\leq j_r$
 and $\dim(U_i\bigcap V_{k-r})= u_i-r  $ for $j_r+1\leq i\leq s.$
 Hence $$m_{G_i}(V_{k-r})= 0$$
 for $1\leq i\leq j_r$ and  $m_{G_i}(V_{k-r})= q^{u_i-r}-1$ for $j_r+1\leq i\leq s.$

Hence $\min\{\sum_{i=1}^s m_{G_i}(V) \,| V\in {\rm SUB}^{k-r}(\mathbb{F}_q^{k})\}
= \sum_{i=j_r+1}^s \frac{q^{ u_i-r}-1}{q-1}$
and $$ d_{r}(C)=t\frac{q^{k}-q^{k-r}}{q-1}-\sum_{i=1}^s\frac{q^{u_i}-1}{q-1}+
\sum_{i=j_r+1}^s \frac{q^{ u_i-r}-1}{q-1}  .$$

{\bf(d)}
Note that $A_j^{r}(C)=0$ for $1\leq j<d_r(C).$
Then we assume $d_r(C)\leq j\leq n.$
 We use Equality~(\ref{oo31}), i.e.,
\begin{equation*}
  w(U)=t\frac{q^{k}-q^{k-r}}{q-1}-\sum_{i=1}^s\frac{q^{u_i}-1}{q-1}+\sum_{i=1}^s m_{G_i}(V^{\bot})
\end{equation*}
to determine the $r$-SSWD of $C$.
Let $m_{r,j}=j-t\frac{q^{k}-q^{k-r}}{q-1}+\sum_{i=1}^s\frac{q^{u_i}-1}{q-1}.$
Then
\begin{eqnarray*}
  A_j^{r}(C) &=& |\{ U\in {\rm SUB}^{r}(C)\,|\, w(U)=j\}| \\
   &=& |\{ V\in {\rm SUB}^{r}(\mathbb{F}_q^k)\,|\,     m_{r,j}=\sum_{i=1}^s m_{G_i}(V^{\bot})\}|\\
&=& |\{ V\in {\rm SUB}^{k-r}(\mathbb{F}_q^k)\,|\,     m_{r,j}=\sum_{i=1}^s m_{G_i}(V)\}|\\
&=& |\{ V\in {\rm SUB}^{k-r}(\mathbb{F}_q^k)\,|\,     m_{r,j}=\sum_{i=1}^s   \frac{q^{ \dim(U_i\bigcap V)}-1}{q-1}\}|\\
&=& \sum_{( v_1,v_2,\cdots, v_s,k-r)\in \mathbf{Z}(q, m_{r,j})}  |\{ V\in {\rm SUB}^{k-r}(\mathbb{F}_q^k)\,|\,  \dim(U_i\bigcap V)=v_i \text{ for }  1\leq i\leq s \}|\\
&=&\sum_{( v_1,v_2,\cdots, v_s,k-r)\in \mathbf{Z}(q, m_{r,j})} N^{u_1, u_2, \cdots, u_s,k}_{v_1,v_2,\cdots, v_s,k-r} ,
\end{eqnarray*}

where  { \small $$ \mathbf{Z}(q, m_{r,j}) = \Biggm{\{}( v_1,\cdots, v_s,k-r)\in \mathbb{Z}^{s+1}\,\Biggm{|}\, 
\left\{
  \begin{array}{ll}
    0\leq v_i-v_{i-1}\leq u_i-u_{i-1}, & \hbox{$1\leq i\leq s$;} \\
    \max\{0,\,u_i-r\}\leq v_i\leq \min\{u_i,\,k-r\}, & \hbox{$1\leq i\leq s$;} \\
    m_{r,j}=\sum_{i=1}^{s} \frac{q^{v_i}-1}{q-1}.& \hbox{}
  \end{array}
\right.
  \Biggm{\}} .$$}

Assume $v_0=0$ and  $ v_{s+1}=k-r.$ By Lemma~\ref{num}, we have that
$$A_j^{r}(C) =\sum_{( v_1,v_2,\cdots, v_s,v_{s+1})\in \mathbf{Z}(q, m_{r,j}) }  \prod_{i=1}^{s+1}q^{(u_{i-1}-v_{i-1})(v_{i}-v_{i-1})}\dst{u_{i}-u_{i-1}}{v_{i}-v_{i-1}}_q.$$
\end{proof}

\begin{Remark}
When $1\leq s\leq q-1,$ the size of the set $\mathbf{Z}(q, m_{r,j})$ in Theorem~\ref{xxd2} is less than two by using the condition $m_{r,j}=\sum_{i=1}^{s} \frac{q^{v_i}-1}{q-1}$ and the uniqueness of the
 $q$-adic expansion.
\end{Remark}

\begin{Corollary}\label{2a}
Assume the notation is as given above and $s=1$.
For an integer $r$ with $1\leq r\leq k$,
then 
$$
    |\{j\,|\, A_j^{r}(C) \neq 0\}| \leq \min\{ u_1+1,\,r+1,\,k-r+1,\,k-u_1+1\} .
$$

\end{Corollary}

\begin{proof}
Assume $s=1.$
By Theorem~\ref{xxd2}, we know that
$$ |\{j\,|\, A_j^{r}(C) \neq 0\}| \leq |\{\mathbf{Z}(q, m_{r,j})\,|\, \mathbf{Z}(q, m_{r,j})\neq \emptyset\}|.$$
Then we enumerate the integer $v_1$ satisfying  $  \max\{0,\,u_1-r\}\leq v_1\leq \min\{u_1,\,k-r\}.$
Hence we have that $$ |\{j\,|\, A_j^{r}(C) \neq 0\}| \leq\min\{ u_1+1,\,r+1,\,k-r+1,\,k-u_1+1\} .$$
\end{proof}

By Corollary~\ref{2a}, we know that the linear codes constructed in Theorem~\ref{xxd2} have few subcode support weights with dimension $r$, when $s=1$ and $\min\{ u_1+1,\,r+1,\,k-r+1,\,k-u_1+1\}$ is small.

\begin{Example}\label{ee12p}
Assume $q=s=u_1=2,$ $u_2=3$ and $k=4$ in Theorem~\ref{xxd2}.
Let  
{ \tiny $$S_{2,4}=\left(\begin{array}{ccccc ccccc ccccc}
    1&0&1&0&1&0&1 &0&1&0&1&0&1&0&1\\
  0&1&1&0&0&1&1 &0&0&1&1&0&0&1&1\\
   0&0&0&1&1&1&1 &0&0&0&0&1&1&1&1\\
   0&0&0&0&0 &0&0 &1&1&1&1&1&1&1&1
\end{array}\right),$$
$$ G_1=\left(\begin{array}{ccc }
      1&0 &1\\
      0&1&1\\
      0&0&0\\
      0&0&0
\end{array}\right) \text{ and } G_2=\left(\begin{array}{cc ccccc}
      1&0&1&0&1&0&1\\
      0&1&1&0&0&1&1 \\
      0&0&0&1&1&1&1\\
      0&0&0&0&0 &0&0
\end{array}\right) .$$}
Let $C_1$ be the $[20,4,10]_2$-linear code with the generator matrix $[S_{2,4}\backslash G_1,S_{2,4}\backslash G_2]$,
and let $C_2$ be the $[35,4,18]_2$-linear code with the generator matrix $[S_{2,4}\backslash G_1,S_{2,4}\backslash G_2,S_{2,4}].$
The linear codes $C_1$ and $C_2$ are Griesmer codes.
By Magma and Theorem~\ref{xxd2}, the subcode support weight distributions of $C_1$ and $C_2$ are listed in Table~\ref{tr}.
\begin{table}[htbp]

{ \small
\caption{Parameters of  $C_1$ and $C_2$ in Example~\ref{ee12p}}
\label{tr}
\center
\begin{tabular}{c |c}

parameters of $C_1$& parameters of $C_2$\\ 
$[20,4,10]_2$&$[35,4,18]_2$ \\ \hline
$1$-SSWD of $C_1$&$1$-SSWD of $C_2$\\
\{[10,12],[12,2],[16,1]\}&\{[18,12],[20,2],[24,1]\}\\ \hline
$2$-SSWD of $C_1$&$2$-SSWD of $C_2$\\
\{[15,16],[16,12],[18,6],[20,1]\}& \{[27,16],[28,12],[30,6],[32,1]\}\\ \hline
$3$-SSWD of $C_1$&$3$-SSWD of $C_2$\\
 \{[18,8],[19,4],[20,3]\}& \{[32,8],[33,4],[34,3]\}\\

\end{tabular}

}
\end{table}

\end{Example}

\subsection{The subcode support weight distributions  of $r$-Griesmer codes for $r\ge 2$}

In the following theorem, we construct an infinite family of distance-optimal $r$-Griesmer codes for $r\ge 2$, and we determine the subcode support weight distributions of those codes.

\begin{Theorem}\label{xxd}
Assume the notation is as given above. For integers $$0=u_0<1= u_1< u_2< \cdots< u_s< u_{s+1}=k$$ and $2\leq s\leq t+1,$
 there exists  an $[n,k,d]_q$-linear code $C$ with  $$n=t\frac{q^k-1}{q-1}-\sum_{i=2}^s\frac{q^{u_i}-1}{q-1}+1\,\,\,\text{ and
} \,\,\, d=tq^{k-1}-\sum_{i=2}^sq^{u_i-1}. $$
And the linear code $C$ satisfies the following properties:

\begin{description}

     \item[(a)] 
  The weight distribution of the linear code $C$ is completely determined in the following table:
$$\begin{array}{c|c}
                                         \text{weight} & \text{multiplicity} \\
                                           \hline
                                         0 & 1 \\
                                        d&q^{k-1}-q^{k-u_2}\\
                                          d+1& q^{k}-q^{k-1} \\
                                          d+q^{u_2-1}& q^{k-u_2}-q^{k-u_3}\\
                                          d+q^{u_2-1}+q^{u_3-1}& q^{k-u_3}-q^{k-u_4}\\
                                         \vdots&\vdots\\
                                          d+\sum_{i=2}^{s-1}q^{u_i-1}&q^{k-u_{s-1}}-q^{k-u_s}\\
                                          d+\sum_{i=2}^sq^{u_i-1}&q^{k-u_s}-1\\
                                          \hline
                                        \end{array}  .$$
  \item[(b)]    The $r$-GHW of $C$ satisfies $$
   d_r(C) = \left\{
              \begin{array}{ll}
                t\frac{q^{k}-q^{k-r}}{q-1}-\sum_{i=2}^s\frac{q^{u_i}-q^{u_i-r}}{q-1}, & \hbox{$1\leq r<u_2$;}\\
               t\frac{q^{k}-q^{k-r}}{q-1}-\sum_{i=2}^s\frac{q^{u_i}-1}{q-1}
+\sum_{i=j_r+1}^s\frac{q^{u_i-r}-1}{q-1}+1 , & \hbox{$ u_2\leq r\leq k$.}
              \end{array}
            \right.
$$

 where $j_r=\max\{ j\,|\,2\leq j\leq s \text{ and } u_{j}\leq r\}.$

 \item[(c)] Assume $r=u_2$. The linear code $C$ is a $r$-Griesmer code, an almost Griesmer code and a distance-optimal code.

     \item[(d)]  Assume $v_0=0$ and $ v_{s+1}=k-r$. 
The $r$-SSWD of $C$ is $$A_j^{r}(C) =\sum_{( v_1,v_2,\cdots, v_s,v_{s+1})\in \mathbf{Z}(q, m_{r,j}) }  \prod_{i=1}^{s+1}q^{(u_{i-1}-v_{i-1})(v_{i}-v_{i-1})}\dst{u_{i}-u_{i-1}}{v_{i}-v_{i-1}}_q,$$
where  $m_{r,j}=j-t\frac{q^{k}-q^{k-r}}{q-1}+\sum_{i=2}^s\frac{q^{u_i}-1}{q-1}-1$ and 
 { \small $$ \mathbf{Z}(q, m_{r,j}) = \Biggm{\{}( v_1,\cdots, v_s,k-r)\in \mathbb{Z}^{s+1}\, \Biggm{|}\, 
\left\{
  \begin{array}{ll}
    0\leq v_i-v_{i-1}\leq u_i-u_{i-1}, & \hbox{$1\leq i\leq s$;} \\
    \max\{0,\,u_i-r\}\leq v_i\leq \min\{u_i,\,k-r\}, & \hbox{$1\leq i\leq s$;} \\
    m_{r,j}=\sum_{i=2}^{s} \frac{q^{v_i}-1}{q-1}-\frac{q^{v_1}-1}{q-1} .& \hbox{}
  \end{array}
\right.
  \Biggm{\}}.$$}

\end{description}
\end{Theorem}

\begin{proof}
Let $\mathbf{e}_i\in \mathbb{F}_q^k$ be the vector with all $0$s except for a $1$ in the $i$th coordinate.
Assume $U_i$ is the linear subspace generated by $\{\mathbf{e}_{1},\,\mathbf{e}_{2},\cdots ,\mathbf{e}_{u_i}\}$ for $1\leq i\leq s+1.$
Let $G_i$ be  the $k \times \frac{q^{u_i}-1}{q-1}$ submatrix  of $S_{q,k}$ such that each column of $G_i$ is in $U_i$.  
Let $C$ be the $[n,k]_q$-linear code with the generator matrix 
\begin{equation}\label{de}
  G=[G_1,S_{q,k}\backslash G_2, S_{q,k}\backslash G_3,\cdots,S_{q,k}\backslash G_s,\underbrace{S_{q,k},\cdots,S_{q,k}}_{t-s+1}],
\end{equation}
where $n=t\frac{q^k-1}{q-1}-\sum_{i=2}^s\frac{q^{u_i}-1}{q-1}+1$.


Note that, for any subspace $U\in {\rm SUB}^{r}(C)$, there exists a subspace $V\in \mathrm{SUB}^r (\mathbb{F}_q^k)$  such that $U=\{ \mathbf{y}G\,|\,  \mathbf{y}\in V\}.$
By Lemma~\ref{weight}, we know that the support weight of $U$ is 
\begin{equation}\label{uu6}
 w(U)= n-m_{G}( V^{\bot}),
\end{equation}
where $n=t\frac{q^k-1}{q-1}-\sum_{i=2}^s\frac{q^{u_i}-1}{q-1}+1$ and $\dim(V^{\bot})=k-r.$
Assume $\tilde{G}$ is the $k\times t\frac{q^k-1}{q-1}$ matrix $[\underbrace{S_{q,k}, S_{q,k},\cdots,S_{q,k}}_t].$
By the definition of the function $m_{G}$, we have that
\begin{equation}\label{oor6}
  m_{\tilde{G}}(V^{\bot})=m_{G}(V^{\bot})+\sum_{i=2}^s m_{G_i}(V^{\bot}) -m_{G_1}(V^{\bot}) .
\end{equation}

By the definition of $S_{q,k}$,
we know that $m_{\tilde{G}}(V^{\bot})=t\frac{q^{k-r}-1}{q-1} .$

By Equalities~(\ref{uu6}) and (\ref{oor6}), we have that
\begin{equation}\label{oo1}
  w(U)=t\frac{q^{k}-q^{k-r}}{q-1}-\sum_{i=2}^s\frac{q^{u_i}-1}{q-1}+1+\sum_{i=2}^s m_{G_i}(V^{\bot})-m_{G_1}(V^{\bot}).
\end{equation}

Note that $\dim(U_i)=u_i$, $\dim(V^{\bot})=k-r$ and $U_i+V^{\bot}\subseteq \mathbb{F}_q^k.$
Then we know that $$\max\{0,\,u_i-r\}\leq \dim(U_i\bigcap V^{\bot})\leq \min\{u_i,\,k-r\}.$$

Hence we get that
 $$m_{G_i}(V^{\bot})\ge \max\{0,\,\frac{q^{u_i-r}-1}{q-1}\}$$
 for $1\leq i\leq s$ and $m_{G_1}(V^{\bot})\in  \{0,1\}.$

{\bf (a)} 
The proof of Statement {\bf (a)} is similar to the proof of Theorem~\ref{xxd2}.


{\bf (b)} 
Assume $1=u_1\leq r<u_2.$
Then $m_{G_i}(V^{\bot})\ge \frac{q^{u_i-r}-1}{q-1}$ for $1\leq i\leq s$ and 
$$\sum_{i=2}^s m_{G_i}(V^{\bot})-m_{G_1}(V^{\bot})\ge \sum_{i=2}^s\frac{q^{u_i-r}-1}{q-1}-1.$$

 Assume that $V_r$ is the linear subspace generated by $\{\mathbf{e}_{2},\,\mathbf{e}_{3},\cdots,\mathbf{e}_{r+1}\}.$
 Then $V_r^{\bot}$ is the linear subspace generated by $\{\mathbf{e}_{1},\,\mathbf{e}_{r+2},\mathbf{e}_{r+3},\cdots,\mathbf{e}_{k}\}.$
 Note that $\dim(U_i\bigcap V_r^{\bot})= u_i-r  $ for $2\leq i\leq s$
 and $\dim(U_i\bigcap V_r^{\bot})=1$.
  Hence $$\sum_{i=2}^s m_{G_i}(V_r^{\bot})-m_{G_1}(V_r^{\bot})= \sum_{i=2}^s\frac{q^{u_i-r}-1}{q-1}-1.$$
 By Equality~(\ref{oo1}), we have that
 $d_r(C) = t\frac{q^{k}-q^{k-r}}{q-1}-\sum_{i=2}^s\frac{q^{u_i}-q^{u_i-r}}{q-1}$
 for $1\leq r<u_2$.

Assume $u_2\leq r\leq k.$
Let $j_r=\max\{ j\,|\,2\leq j\leq s \text{ and } u_{j}\leq r\}.$
Note that $\dim(V)=r$ and $\dim(V^{\bot})=k-r.$
Then $$\dim(U_i \bigcap V^{\bot})\ge 0$$ for $1\leq i\leq j_r$
and $\dim(U_i \bigcap V^{\bot})\ge u_i-r$ for $j_r+1\leq i\leq s.$
If $\dim(U_2\bigcap V^{\bot})= 0$, then $\dim(U_1\bigcap V^{\bot})= 0.$
Hence
$$\sum_{i=2}^{j_r} m_{G_i}(V_r^{\bot})-m_{G_1}(V^{\bot})\ge 0
  \text{ and }\sum_{i=j_r+1}^{s} m_{G_i}(V_r^{\bot})\ge \sum_{i=j_r+1}^{s}\frac{q^{u_i-r}-1}{q-1}.$$

 Assume that $\tilde{V } _r$ is the linear subspace generated by $\{\mathbf{e}_{1},\,\mathbf{e}_{2},\cdots,\mathbf{e}_{r}\}.$
 Then $\tilde{V}_r^{\bot}$ is the linear subspace generated by $\{\mathbf{e}_{r+1},\,\mathbf{e}_{r+2},\cdots,\mathbf{e}_{k}\}.$
 Note that $\dim(U_i\bigcap \tilde{V}_r^{\bot})= 0  $ for $1\leq i\leq j_r$
 and $\dim(U_i\bigcap \tilde{V}_r^{\bot})= u_i-r  $ for $j_r+1\leq i\leq s,$
 where $\dim(U_i)=u_i$.
  Hence  $$\sum_{i=2}^s m_{G_i}(\tilde{V}_r^{\bot})-m_{G_1}(\tilde{V}_r^{\bot})= \sum_{i=j_r+1}^s\frac{q^{u_i-r}-1}{q-1}.$$
 By Equality~(\ref{oo1}), we have that $d_r(C)=t\frac{q^{k}-q^{k-r}}{q-1}-\sum_{i=2}^s\frac{q^{u_i}-q^{u_i-r}}{q-1}
+\sum_{i=j_r+1}^s\frac{q^{u_i-r}-1}{q-1}+1,$
 where $j_r=\max\{ j\,|\,2\leq j\leq s \text{ and } u_{j}\leq r\}.$

 {\bf (c)}
It is easy to prove that the linear code $C$ is a $u_2$-Griesmer code and an almost Griesmer code by Statement  {\bf (b)}.
Suppose there exists an $[n, k,\bar{d}]_q$-linear code $\bar{C}$ such that $\bar{d}\ge d+1.$
Since $d$ is divisible by $q$, we have that $\lceil \frac{\bar{d}}{q}\rceil\ge  \lceil \frac{d+1}{q}\rceil=\lceil\frac{d}{q}\rceil+1.$
Then $$\sum_{i=0}^{k-1}\lceil \frac{\bar{d}}{q^i}\rceil \ge  \sum_{i=0}^{k-1}\lceil \frac{d}{q^i}\rceil+2=n+1,$$
which is a contradiction to the Griesmer bound for the minimum Hamming
weights of linear codes. Therefore $C$ is a distance-optimal code.
 
{\bf (d)} 
Note that $A_j^{r}(C)=0$ for $1\leq j<d_r(C).$
Then we assume $d_r(C)\leq j\leq n.$
 We use Equality~(\ref{oo1})
to determine the $r$-SSWD $[A_1^{r}(C), A_2^{r}(C),\cdots, A_n^{r}(C)]$ of $C$.
Let $$ m_{r,j}=j-t\frac{q^{k}-q^{k-r}}{q-1}+\sum_{i=2}^s\frac{q^{u_i}-1}{q-1}-1.$$
Since $m_{G_i}(V)=\frac{q^{\dim(U_i\bigcap V)}-1}{q-1}$ for $1\leq i\leq s$, we have that 
\begin{eqnarray*}
  A_j^{r}(C) &=& |\{ U\in {\rm SUB}^{r}(C)\,|\,w(U)=j\}| \\
   &=& |\{ V\in {\rm SUB}^{r}(\mathbb{F}_q^k)\,|\,     m_{r,j}=\sum_{i=2}^s m_{G_i}(V^{\bot})-m_{G_1}(V^{\bot})\}|\\
&=& |\{ V\in {\rm SUB}^{k-r}(\mathbb{F}_q^k)\,|\,     m_{r,j}=\sum_{i=2}^s m_{G_i}(V)-m_{G_1}(V)\}|\\
&=& \sum_{( v_1,v_2,\cdots, v_s,k-r)\in \mathbf{Z}(q, m_{r,j})}  |\{ V\in {\rm SUB}^{k-r}(\mathbb{F}_q^k)\,|\,  \dim(U_i\bigcap V)=v_i \text{ for }  1\leq i\leq s \}|\\
&=&\sum_{( v_1,v_2,\cdots, v_s,k-r)\in \mathbf{Z}(q, m_{r,j})} N^{u_1, u_2, \cdots, u_s,k}_{v_1,v_2,\cdots, v_s,k-r} ,
\end{eqnarray*}

where 
 { \small $$ \mathbf{Z}(q, m_{r,j}) = \Biggm{\{}( v_1,\cdots, v_s,k-r)\in \mathbb{Z}^{s+1}\, \Biggm{|}\, 
\left\{
  \begin{array}{ll}
    0\leq v_i-v_{i-1}\leq u_i-u_{i-1}, & \hbox{$1\leq i\leq s$;} \\
    \max\{0,\,u_i-r\}\leq v_i\leq \min\{u_i,\,k-r\}, & \hbox{$1\leq i\leq s$;} \\
    m_{r,j}=\sum_{i=2}^{s} \frac{q^{v_i}-1}{q-1}-\frac{q^{v_1}-1}{q-1} .& \hbox{}
  \end{array}
\right.
  \Biggm{\}  }.$$}

Assume $v_0=0$ and  $ v_{s+1}=k-r.$ By Lemma~\ref{num}, we have that
$$A_j^{r}(C) =\sum_{( v_1,v_2,\cdots, v_s,v_{s+1})\in \mathbf{Z}(q, m_{r,j}) }  \prod_{i=1}^{s+1}q^{(u_{i-1}-v_{i-1})(v_{i}-v_{i-1})}\dst{u_{i}-u_{i-1}}{v_{i}-v_{i-1}}_q.$$

 \end{proof}

\begin{Remark}
If the set $\mathbf{Z}(q, m_{r,j})$ of Statement \textbf{(c)} in Theorem~\ref{xxd} is empty, then
 $A_j^{r}(C) =0$.
 When $1\leq s\leq q-1,$ the size of the set $\mathbf{Z}(q, m_{r,j})$ in Theorem~\ref{xxd} is less than two by using the condition $m_{r,j}=\sum_{i=2}^{s} \frac{q^{v_i}-1}{q-1}-\frac{q^{v_1}-1}{q-1}$ and the uniqueness of the $q$-adic expansion.
\end{Remark}

\begin{Corollary}\label{2}
Assume the notation is as given above and $s=2$.
For an integer $r$ with $1\leq r\leq k$,
then $$ |\{j\,|\, A_j^{r}(C) \neq 0\}| \leq \min\{ 2u_2,\,2(r+1),\,2(k-r+1),\,2(k-u_2+1)\}.$$ 
\end{Corollary}

\begin{proof}
By Theorem~\ref{xxd}, we know that
$$ |\{j\,|\, A_j^{r}(C) \neq 0\}| \leq |\{\mathbf{Z}(q, m_{r,j})\,|\, \mathbf{Z}(q, m_{r,j})\neq \emptyset\}|.$$
Then we enumerate all the integers $v_1$ and $v_2$ satisfying $0\leq v_1\leq 1$, $0\leq v_2-v_{1}\leq u_2-1$ and $  \max\{0,\,u_2-r\}\leq v_2\leq \min\{u_2,\,k-r\}.$
Hence we have that $$ |\{j\,|\, A_j^{r}(C) \neq 0\}| \leq \min\{ 2u_2,\,2(r+1),\,2(k-r+1),\,2(k-u_2+1)\} .$$
\end{proof}

By Corollary~\ref{2}, we know that the linear codes constructed in Theorem~\ref{xxd} have few subcode support weights with dimension $r$, when $s=2$ and $\min\{ 2u_2+1,\,2(r+1),\,2(k-r+1),\,2(k-u_2+1)\}$ is small.

\begin{Example}\label{ee124}
Assume $s=2$ and $q=k-1=3$ in Theorem~\ref{xxd}.
Let  
{ \tiny$$S_{3,4}=\left(\begin{array}{c|ccc |c ccccc ccc|c}
1&0&1 &2&0&1  &2&0 &1&2&0&1&2& \cdots\\
0&1&1 &1&0&0  &0&1 &1&1&2&2&2&\cdots\\
0&0&0&0&1&1  &1&1  &1&1&1&1&1&\cdots\\
0&0&0&0&0&0  &0&0  &0&0&0&0&0&\cdots
\end{array}\right), G_1=\left(\begin{array}{c}
      1\\
      0 \\
      0\\
      0
\end{array}\right) , $$
$$
 G_2=\left(\begin{array}{cc cc}
      1&0&1&2\\
      0&1&1&1 \\
      0&0&0&0\\
      0&0&0&0
\end{array}\right) \text{ and } G_{3}=\left(\begin{array}{cccc c ccccc ccc}
1&0&1 &2&0&1  &2&0 &1&2&0&1&2\\
0&1&1 &1&0&0  &0&1 &1&1&2&2&2\\
0&0&0&0&1&1  &1&1  &1&1&1&1&1\\
0&0&0&0&0&0  &0&0  &0&0&0&0&0
\end{array}\right).$$
}
Let $C_1$ be the $[37,4,24]_3$-linear code with the generator matrix $$[G_1,S_{4,3}\backslash G_2],$$
which is on the condition of $u_2=2$ in Theorem~\ref{xxd}. And let  $C_2$ be the $[28,4,18]_3$-linear code with the generator matrix $$[G_1,S_{4,3}\backslash G_3],$$ which is on the condition of $u_2=3$ in Theorem~\ref{xxd}.
The linear code $C_1$ is a $2$-Griesmer code, and  the linear code $C_2$ is a  $3$-Griesmer code. 
The linear codes $C_1$ and $C_2$ are almost Griesmer codes and distance-optimal codes.
By Magma and Theorem~\ref{xxd}, the subcode support weight distributions of $C_1$ and $C_2$ are listed in Table~\ref{t23}.

\begin{table}[htbp]
\caption{Parameters of $C_1$ and $C_2$ in Example~\ref{ee124}}
\label{t23}

\center
{ \small
\begin{tabular}{c |c}

parameters of $C_1$&parameters of $C_2$\\ 
$[37,4,24]_3$&$[28,4,18]_3$ \\ \hline
$1$-SSWD of $C_1$&$1$-SSWD of $C_2$\\
\{[24,9],[25,27],[27,4]\}&\{[18,12],[19,27],[27,1]\}\\ \hline
$2$-SSWD of $C_1$&$2$-SSWD of $C_2$\\
\{[33,93],[34,36],[36,1]\}& \{[24,9],[25,108],[27,4],[28,9]\}\\ \hline
$3$-SSWD of $C_1$&$3$-SSWD of $C_2$\\
 \{[36,37],[37,3]\}& \{[27,28],[28,12]\}\\

\end{tabular}
}
\end{table}

\end{Example}

\section{The subcode support weight distributions of modified Solomon-Stiffler codes}
In this section, we construct a family of distance-optimal $r$-Griesmer codes by modified Solomon-Stiffler codes, and determine their subcode support weight distributions. And $r$-Griesmer codes in this section are different from codes in the last section.

When we determine the subcode support weight distributions of linear codes constructed in this section, we need the following lemma. For integers $0< u_1< u_2<  u_3< k$, assume $U_i\in{\rm SUB}^{r}(\mathbb{F}_q^k) $ 
such that $\dim(U_i)=u_i.$ 
Let $$M^{u_1, u_2, u_3,k}_{v_1, v_2, v_3,r}=|\{ V\in  {\rm SUB}^{r}(\mathbb{F}_q^k)\,|\, 
 \dim(U_i\bigcap V)=v_i \text{ for  } 1\leq i\leq 3\}|,$$
where $U_1=U_2\bigcap U_3$.

\begin{lem}\label{xx14}
Assume the notation is as given above. Then 
\begin{description}
\item[(a)] If $u_1=v_1=0$ and $u_2+ u_3=k,$ 
then $$M^{0, u_2, u_3,u_2+ u_3}_{0, v_2, v_3,v_2+v_3+t}=\left\{
         \begin{array}{ll}
           \dst{u_2-v_2}{t}_q\dst{u_3-v_3}{t}_q\dst{u_2}{v_2}_q\dst{u_3}{v_3}_q\prod_{i=0}^{t-1}(q^{t}-q^{i}), & \hbox {$t>0$;} \\
             \dst{u_2}{v_2}_q\dst{u_3}{v_3}_q, & \hbox{$ t=0$.}
             \end{array}
          \right.
$$

  \item[(b)] If $u_1=v_1=0,$ then $$M^{0, u_2, u_3,k}_{0, v_2, v_3,r}
=\sum_{t=0}^{\tilde{t}}q^{(u_2+u_3-v_2-v_3-t)(r-v_2-v_3-t)}M^{0, u_2, u_3,u_2+ u_3}_{0, v_2, v_3,v_2+ v_3+t} \dst{k-u_2-u_3}{r-v_2-v_3-t}_q ,$$
where $\tilde{t}=\min\{u_2-v_2, u_3-v_3,r-v_2-v_3\}$.

  \item[(c)] If $u_1=1$ and $v_1=1,$ then 
$$M^{1, u_2, u_3,k}_{1, v_2, v_3,r}=M^{0, u_2-1, u_3-1,k-1}_{0, v_2-1, v_3-1,r-1}.$$
  \item[(d)] 
Assume $r<k$. If $u_1=1$ and $v_1=0,$ then 
$$M^{1, u_2, u_3,k}_{0, v_2, v_3,r}=q^rM^{1, u_2, u_3,k}_{1, v_2+1, v_3+1,r+1}=q^rM^{0, u_2-1, u_3-1,k-1}_{0, v_2, v_3,r}.$$
\end{description}
\end{lem}
\begin{proof}
\textbf{(a)} It is Lemma 4.1 of \cite{SL}.

\textbf{(b)} It is Lemma 4.3 of \cite{SL}.

\textbf{(c)} 
Let $U_1=U_2\bigcap U_3$ and $\dim(U_1)=u_1=1$.
There is a unique map $\phi$ from $\mathbb{F}_q^k$ onto the quotient space $\mathbb{F}_q^k/U_1$
such that $$\phi(\mathbf{x})=\mathbf{x}+U_1.$$
Then $\phi$ induce a bijective map from $\{V\in  {\rm SUB}^{r}(\mathbb{F}_q^k)\,|\, U_1\subseteq V\}$ onto ${\rm SUB}^{r-1}(\mathbb{F}_q^k/U_1)$.
Note that $$\phi(U_2)\bigcap \phi( U_3)=\phi(U_2\bigcap U_3)=\{\mathbf{0}\}$$
and $$\dim(\phi(U_i)\bigcap \phi( V))=\dim(\phi(U_i\bigcap V))=v_i-1$$ for $i=2,3.$
Hence $M^{1, u_2, u_3,k}_{1, v_2, v_3,r}=M^{0, u_2-1, u_3-1,k-1}_{0, v_2-1, v_3-1,r-1}.$

\textbf{(d)} 
 Let $$\Lambda^{1, u_2, u_3,k}_{0, v_2, v_3,r}=\{ V\in  {\rm SUB}^{r}(\mathbb{F}_q^{k})\,|\, 
 \dim(U_i\bigcap V)=v_i \text{ for  } 1\leq i\leq 3\} $$ and 
$$\Lambda^{1, u_2, u_3,k}_{1, v_2+1, v_3+1,r+1}=\{ \tilde{V}\in  {\rm SUB}^{r+1}(\mathbb{F}_q^{k})\,|\, 
 \dim(U_i\bigcap V)=v_i+1 \text{ for  } 1\leq i\leq 3\} .$$
Then 
 $$M^{1, u_2, u_3,k}_{0, v_2, v_3,r}=|\Lambda^{1, u_2, u_3,k}_{0, v_2, v_3,r}|\text{ and 
}M^{1, u_2, u_3,k}_{1, v_2+1, v_3+1,r+1}=|\Lambda^{1, u_2, u_3,k}_{1, v_2+1, v_3+1,r+1} |.$$
We can construct a map $\psi$ from $\Lambda^{1, u_2, u_3,k}_{0, v_2, v_3,r}$ to $\Lambda^{1, u_2, u_3,k}_{1, v_2+1, v_3+1,r+1}$ such that $\psi(V)=V+U_1$ for any $V\in M^{1, u_2, u_3,k}_{0, v_2, v_3,r}.$
It is obvious that $\psi$ is surjective.

For any $\tilde{V}\in M^{1, u_2, u_3,k}_{1, v_2+1, v_3+1,r+1}$, 
we have that $\dim(\tilde{V}\bigcap U_i)=v_i+1$ for $1\leq i\leq 3.$
Since $\dim(\tilde{V})=\dim(V)+1$, we know that $$\dim((U_i\bigcap \tilde{V})\bigcap V )\in \{v_i,v_i+1\}$$ for $1\leq i\leq 3.$
Note that $$U_1\bigcap \tilde{V}=(U_2\bigcap \tilde{V})\bigcap(U_3\bigcap \tilde{V}).$$
Therefore, if $\dim((U_1\bigcap \tilde{V})\bigcap V )=v_1=0,$ then $\dim((U_i\bigcap \tilde{V})\bigcap V )=v_i$ for $1\leq i\leq 3.$
By Lemma~\ref{num}, we have that
\begin{eqnarray*}
  |\psi^{-1}(\tilde{V})| &=& |\{ V\in  {\rm SUB}^{r}(\tilde{V})\,|\, 
 \dim(U_i\bigcap V)=v_i \text{ for  } 1\leq i\leq 3\}| \\
&=& |\{ V\in  {\rm SUB}^{r}(\tilde{V})\,|\, 
 \dim((U_i\bigcap\tilde{V})\bigcap V)=v_i \text{ for  } 1\leq i\leq 3\}| \\
   &=&|\{ V\in  {\rm SUB}^{r}(\tilde{V})\,|\, 
 \dim(U_1\bigcap V)=v_1\}| \\
&=&q^r.
\end{eqnarray*}

Hence $M^{1, u_2, u_3,k}_{0, v_2, v_3,r}=q^rM^{1, u_2, u_3,k}_{1, v_2+1, v_3+1,r+1}=q^rM^{0, u_2-1, u_3-1,k-1}_{0, v_2, v_3,r}$ by Statement \textbf{(b)}.
\end{proof}

By Lemma~\ref{xx14}, the integers $M^{0, u_2, u_3,k}_{0, v_2, v_3,r}$, $M^{1, u_2, u_3,k}_{0, v_2, v_3,r}$ and $M^{1, u_2, u_3,k}_{1, v_2, v_3,r}$ are known. Then we can determine the the subcode support weight distributions of linear codes in the next theorem.

\begin{Theorem}\label{xxd4}
Assume the notation is as given above. For integers $$ 1=u_1< u_2<  u_3<k$$ 
such that $u_2+u_3\leq k+1,$
 there exists  an $[n,k,d]_q$-linear code $C$ with  $$n=\frac{q^k-1}{q-1}-\frac{q^{u_2}-1}{q-1}-\frac{q^{u_3}-1}{q-1}+1\,\,\,\text{ and
} \,\,\, d=q^{k-1}-q^{u_2-1}-q^{u_3-1}. $$
And the linear code $C$ satisfies the following properties:

\begin{description}
 \item[(a)] If $u_2+u_3< k+1,$ the weight distribution of the linear code $C$ is completely determined in the following table:
$$\begin{array}{c|c}
\text{weight }& \text{multiplicity} \\\hline
                                        0 & 1 \\
                                      d&q^{k-1}-q^{k-u_2}-q^{k-u_3}+q^{k-u_2-u_3+1}\\  
                 d+1& q^{k}-q^{k-1} \\
  d+q^{u_2-1}&q^{k-u_2}-q^{k-u_2-u_3+1}\\
 d+q^{u_3-1}&q^{k-u_3}-q^{k-u_2-u_3+1}\\
d+q^{u_2-1}+q^{u_3-1}&q^{k-u_2-u_3+1}-1\\
\hline
                                        \end{array}  .$$

\item[(b)] If $u_2+u_3= k+1,$ the weight distribution of the linear code $C$ is completely determined in the following table:
$$\begin{array}{c|c}
\text{weight }& \text{multiplicity} \\\hline
                                        0 & 1 \\
                                      d&q^{k-1}-q^{k-u_2}-q^{k-u_3}+q^{k-u_2-u_3+1}\\  
                 d+1& q^{k}-q^{k-1} \\
  d+q^{u_2-1}&q^{k-u_2}-q^{k-u_2-u_3+1}\\
 d+q^{u_3-1}&q^{k-u_3}-q^{k-u_2-u_3+1}\\
\hline
                                        \end{array}  .$$

  \item[(c)]    The $r$-GHW of $C$ satisfies $$d_r(C) = \left\{
              \begin{array}{ll}
                \frac{q^{k}-q^{k-r}-(q^{u_2}-q^{u_2-r})-(q^{u_3}-q^{u_3-r})}{q-1}, & \hbox{$1\leq r<u_2$;}\\
                \frac{q^{k}-q^{k-r}-(q^{u_2}-1)-(q^{u_3}-q^{u_3-r})}{q-1}+1 , & \hbox{$ u_2\leq r< u_3$;}\\
\frac{q^{k}-q^{k-r}-(q^{u_2}-1)-(q^{u_3}-1)}{q-1} +1, & \hbox{$ u_3\leq r\leq  k$.}\\
              \end{array}
            \right.$$

\item[(d)] The $r$-SSWD  of $C$ is $$A_j^{r}(C)=\sum_{(v_1, v_2, v_3)\in\mathbf{Z}(q, m_{r,j})  }M^{1, u_2, u_3,k}_{v_1, v_2, v_3,k-r},$$ where $m_{r,j}=j-\frac{q^{k}-q^{k-r}}{q-1}+\frac{q^{u_2}-1}{q-1}+\frac{q^{u_3}-1}{q-1}-1$
and 
{ \small $$\mathbf{Z}(q, m_{r,j}) =\Biggm{\{}( v_1, v_2, v_3)\in \mathbb{Z}^{3}\,\Biggm{|}\,  \left\{
        \begin{array}{ll}
     \max\{0,\,u_i-r\}\leq v_i\leq \min\{u_i,\,k-r\} , & \hbox{for $ 1\leq i\leq 3$;} \\
      0\leq v_i-v_1\leq u_i-1, & \hbox{for $ 2\leq i\leq 3$;} \\
      v_2+v_3\leq k-r+v_1,& \hbox{ }\\
      m_{r,j}=\sum_{i=2}^{3} \frac{q^{v_i}-1}{q-1}-\frac{q^{v_1}-1}{q-1}. & \hbox{ }
                                                                       \end{array}
                                                                     \right.
\Biggm{\}} .$$}

  \item[(e)] Assume $r=u_2$. The linear code $C$ is a $r$-Griesmer code, an almost Griesmer code and a distance-optimal code.


\end{description}
\end{Theorem}

\begin{proof}
Let $\mathbf{e}_i\in \mathbb{F}_q^k$ be the vector with all $0$s except for a $1$ in the $i$th coordinate.
Assume $U_i$ is the linear subspace generated by $\{\mathbf{e}_{1},\,\mathbf{e}_{2},\cdots ,\mathbf{e}_{u_i}\}$ for $1\leq i\leq 2,$
and $U_3$ is the linear subspace generated by $\{\mathbf{e}_{1}\}\bigcup \{ \mathbf{e}_{k-u_3+2},\,\mathbf{e}_{k-u_3+3},\cdots ,\mathbf{e}_{k}\},$ where $\dim(U_3)=u_3$.
Since $ u_2+u_3\leq k+1$, we have that $U_2 \bigcap U_3=U_1.$

Let $G_i$ be  the $k \times \frac{q^{u_i}-1}{q-1}$ submatrix  of $S_{q,k}$ such that each column of $G_i$ is in $U_i$ for $1\leq i\leq 3$.  And let $[G_2,G_3\backslash G_1]$ be  the $k \times \frac{q^{u_2}-1+q^{u_3}-q}{q-1}$ submatrix  of $S_{q,k}$ such that each column of $[G_2,G_3\backslash G_1]$ is in the subset $U_2\bigcup U_3$.
Then $$G=[S_{q,k}\backslash [G_2,G_3\backslash G_1]]$$ is the $k \times \frac{q^k-q^{u_2}-q^{u_3}+q}{q-1}$ submatrix  of $S_{q,k}$ such that each column of $G$ is in the subset $\mathbb{F}_q^k\backslash (U_2\bigcup U_3)$.
Let $C$ be the $[n,k]_q$-linear code with the generator matrix $G$,
where $$n=\frac{q^k-q^{u_2}-q^{u_3}+q}{q-1}.$$


Note that, for any subspace $U\in {\rm SUB}^{r}(C)$, there exists a subspace $V\in \mathrm{SUB}^r (\mathbb{F}_q^k)$  such that $U=\{ \mathbf{y}G\,|\,  \mathbf{y}\in V\}.$
By Lemma~\ref{weight}, we know that the subcode support weight of $U$ is 
\begin{equation}\label{uu61}
 w(U)= n-m_{G}( V^{\bot}),
\end{equation}
where $n=\frac{q^k-q^{u_2}-q^{u_3}+q}{q-1}$ and $\dim(V^{\bot})=k-r.$
Note that 
\begin{equation}\label{oor61}
  m_{G}(V^{\bot})= \frac{q^{k-r}-1}{q-1}-m_{G_2}(V^{\bot}) -m_{G_3}(V^{\bot}) +m_{G_1}(V^{\bot}) .
\end{equation}

By Equalities~(\ref{uu61}) and (\ref{oor61}), we have that
\begin{equation}\label{oo61}
  w(U)=\frac{q^{k}-q^{k-r}-(q^{u_2}-1)-(q^{u_3}-q) }{q-1}+m_{G_2}(V^{\bot})+ m_{G_3}(V^{\bot})-m_{G_1}(V^{\bot}).
\end{equation}

Note that $\dim(U_i)=u_i$, $\dim(V^{\bot})=k-r$ and $U_i+V^{\bot}\subseteq \mathbb{F}_q^k.$
Then we know that $$\max\{0,\,u_i-r\}\leq \dim(U_i\bigcap V^{\bot})\leq \min\{u_i,\,k-r\}.$$

Hence we get that
 $$m_{G_i}(V^{\bot})\ge \max\{0,\,\frac{q^{u_i-r}-1}{q-1}\}$$
 for $1\leq i\leq 3$ and $m_{G_1}(V^{\bot})\in  \{0,1\}.$

{\bf (a)-\bf (b)}   For any nonzero codeword ${\bf c}\in C$, there exists a unique nonzero vector ${\bf y}\in \mathbb{F}_q^k$ such that ${\bf c}={\bf y}G$. Assume $r=1.$ By Equality~(\ref{oo61}), we have that
\begin{equation}\label{oo66}
  w({\bf c})=\frac{q^{k}-q^{k-1}-(q^{u_2}-1)-(q^{u_3}-q) }{q-1}+m_{G_2}(\langle{\bf y}\rangle^{\bot})+ m_{G_3}(\langle{\bf y}\rangle^{\bot})-m_{G_1}(\langle{\bf y}\rangle^{\bot}).
\end{equation}

Since $\dim(U_i)=u_i$ and $\dim(U_i\bigcap \langle \mathbf{y}\rangle^{\bot})\in \{u_i-1,\,u_i\}$ for $1\leq i\leq s$,
we get that
 $$m_{G_i}(\langle \mathbf{y}\rangle^{\bot})\in \{\frac{q^{u_i-1}-1}{q-1},\,\frac{q^{u_i}-1}{q-1}\}$$
 for $1\leq i\leq 3.$ 
Note that $U_1= U_2\bigcap U_3.$
We proceed to prove that
in the following five cases.

\textbf{Case I:} If $\dim(U_1\bigcap \langle \mathbf{y}\rangle^{\bot})=0$, then 
$$\dim(U_2\bigcap \langle \mathbf{y}\rangle^{\bot})=u_2-1 \text{ and } \dim(U_3\bigcap \langle \mathbf{y}\rangle^{\bot})=u_3-1.$$
Hence $m_{G_2}( \langle \mathbf{y}\rangle^{\bot})+ m_{G_3}( \langle \mathbf{y}\rangle^{\bot})-m_{G_1}( \langle \mathbf{y}\rangle^{\bot})=\sum_{i=2}^3  \frac{q^{u_i-1}-1}{q-1}$
and $$  w({\bf c})=q^{k-1}-q^{u_2-1}-q^{u_3-1}+1$$
by Equality~(\ref{oo66}).
And we know that \begin{eqnarray*}
&&   |\{{\bf y}\in \mathbb{F}_q^k\,|\,  \mathbf{y}\neq \mathbf{0}\text{ and } U_1 \subsetneq\langle \mathbf{y}\rangle^{\bot}\}|\\                          
&=& q^k-1-  |\{{\bf y}\in \mathbb{F}_q^k\,|\,  \mathbf{y}\neq \mathbf{0}\text{ and } U_1 \subseteq \langle \mathbf{y}\rangle^{\bot}\}|\\
                          &=&  q^k-1- |\{{\bf y}\in \mathbb{F}_q^k\,|\,  \mathbf{y}\neq \mathbf{0}\text{ and }
                          y_1=0\}|\\
                          &=& q^{k}-q^{k-1}.
                        \end{eqnarray*}

The rest four cases are similar. Assume $d=q^{k-1}-q^{u_2-1}-q^{u_3-1}$, we list all the five cases as following:
{ \small $$\begin{array}{c|c|c|c}
&\text{condition of }\mathbf{y} &\text{weight of } \mathbf{c}& \text{multiplicity} \\\hline
                                         &\mathbf{y}=\mathbf{0}&0 & 1 \\\hline
                                        
                 \textbf{Case I}                         &\dim(U_1\bigcap \langle \mathbf{y}\rangle^{\bot})=0&d+1& q^{k}-q^{k-1} \\\hline

 &\dim(U_1\bigcap \langle \mathbf{y}\rangle^{\bot})=1&&\\
\textbf{Case II}&\dim(U_2\bigcap \langle \mathbf{y}\rangle^{\bot})=u_2-1       &d&q^{k-1}-q^{k-u_2}-q^{k-u_3}+q^{k-u_2-u_3+1}\\
&\dim(U_3\bigcap \langle \mathbf{y}\rangle^{\bot})=u_3-1&&\\
                                          \hline

&\dim(U_1\bigcap \langle \mathbf{y}\rangle^{\bot})=1&&\\
\textbf{Case III}&\dim(U_2\bigcap \langle \mathbf{y}\rangle^{\bot})=u_2 -1      &d+q^{u_3-1}&q^{k-u_3}-q^{k-u_2-u_3+1}\\
&\dim(U_3\bigcap \langle \mathbf{y}\rangle^{\bot})=u_3&&\\
                                          \hline           
                              
&\dim(U_1\bigcap \langle \mathbf{y}\rangle^{\bot})=1&&\\
\textbf{Case IV}&\dim(U_2\bigcap \langle \mathbf{y}\rangle^{\bot})=u_2       &d+q^{u_2-1}&q^{k-u_2}-q^{k-u_2-u_3+1}-1\\
&\dim(U_3\bigcap \langle \mathbf{y}\rangle^{\bot})=u_3-1&&\\
                                          \hline

\textbf{Case V}&\dim(U_2\bigcap \langle \mathbf{y}\rangle^{\bot})=u_2       &d+q^{u_2-1}+q^{u_3-1}&q^{k-u_2-u_3+1}-1\\
&\dim(U_3\bigcap \langle \mathbf{y}\rangle^{\bot})=u_3&&\\
                                          \hline
                                        \end{array}  .$$}

{\bf (c)} 
Assume $1=u_1\leq  r<u_2.$
Then $m_{G_i}(V^{\bot})\ge \frac{q^{u_i-r}-1}{q-1}$ for $2\leq i\leq 3$ and 
$$m_{G_2}(V^{\bot})+ m_{G_3}(V^{\bot})-m_{G_1}(V^{\bot})\ge \frac{q^{u_2-r}-1}{q-1}+\frac{q^{u_3-r}-1}{q-1}-1.$$

 Assume that $V_r$ is the linear subspace generated by $\{\mathbf{e}_{2},\,\mathbf{e}_{3},\cdots,\mathbf{e}_{r+1}\}.$
 Then $V_r^{\bot}$ is the linear subspace generated by $\{\mathbf{e}_{1},\,\mathbf{e}_{r+2},\mathbf{e}_{r+3},\cdots,\mathbf{e}_{k}\}.$
 Note that $$\dim(U_i\bigcap V_r^{\bot})= u_i-r  $$ for $2\leq i\leq 3$
 and $\dim(U_i\bigcap V_r^{\bot})=1$.
  Hence $$ m_{G_2}(V_r^{\bot}) +m_{G_3}(V_r^{\bot})-m_{G_1}(V_r^{\bot})= \frac{q^{u_2-r}-1}{q-1}+\frac{q^{u_3-r}-1}{q-1}-1.$$
 By Equality~(\ref{oo61}), we have that, for $1\leq r<u_2$, 
 \begin{eqnarray*}
   d_r(C) &=& \min \{w(U)\,\big|\,U\in {\rm SUB}^{r}(C)\}\\
    &=& \frac{q^{k}-q^{k-r}-(q^{u_2}-q^{u_2-r})-(q^{u_3}-q^{u_3-r})}{q-1}.
 \end{eqnarray*}

Analogously, we have that $$d_r(C) = \left\{
              \begin{array}{ll}
                \frac{q^{k}-q^{k-r}-(q^{u_2}-q^{u_2-r})-(q^{u_3}-q^{u_3-r})}{q-1}, & \hbox{$1\leq r<u_2$;}\\
                \frac{q^{k}-q^{k-r}-(q^{u_2}-1)-(q^{u_3}-q^{u_3-r})}{q-1}+1 , & \hbox{$ u_2\leq r< u_3$;}\\
\frac{q^{k}-q^{k-r}-(q^{u_2}-1)-(q^{u_3}-1)}{q-1} +1, & \hbox{$ u_3\leq r\leq  k$.}\\
              \end{array}
            \right.$$

{\bf (d)} Note that $A_j^{r}(C)=0$ for $1\leq j<d_r(C).$
Then we assume $d_r(C)\leq j\leq n.$
 We use Equality~(\ref{oo61}), i.e.,
\begin{equation*}
  w(U)=\frac{q^{k}-q^{k-r}}{q-1}-\frac{q^{u_2}-1}{q-1}-\frac{q^{u_3}-1}{q-1}+1+ m_{G_2}(V^{\bot})+m_{G_3}(V^{\bot})-m_{G_1}(V^{\bot}).
\end{equation*}
to determine the $r$-SSWD  of $C$.
Let $$m_{r,j}=j-\frac{q^{k}-q^{k-r}}{q-1}+\frac{q^{u_2}-1}{q-1}+\frac{q^{u_3}-1}{q-1}-1. $$
Then
\begin{eqnarray*}
  A_j^{r}(C) &=& |\{ U\in {\rm SUB}^{r}(C)\,|\,w(U)=j \}| \\
   &=& |\{ V\in {\rm SUB}^{r}(\mathbb{F}_q^k)\,|\,     m_{r,j}=m_{G_2}(V^{\bot})+
m_{G_3}(V^{\bot})-m_{G_1}(V^{\bot})\}|\\
&=& |\{ V\in {\rm SUB}^{k-r}(\mathbb{F}_q^k)\,|\,     m_{r,j}=m_{G_2}(V)+
m_{G_3}(V)-m_{G_1}(V)\}|.
\end{eqnarray*}
Since $m_{G_i}(V)=\frac{q^{\dim(U_i\bigcap V)}-1}{q-1}$ for $1\leq i\leq 3$, 
we have that 
$$A_j^{r}(C)=\sum_{(v_1, v_2, v_3)\in\mathbf{Z}(q, m_{r,j})  }M^{1, u_2, u_3,k}_{v_1, v_2, v_3,k-r},$$ where 
{ \small $$\mathbf{Z}(q, m_{r,j}) =\Biggm{\{}( v_1, v_2, v_3)\in \mathbb{Z}^{3}\,\Biggm{|}\,  \left\{
        \begin{array}{ll}
     \max\{0,\,u_i-r\}\leq v_i\leq \min\{u_i,\,k-r\} , & \hbox{for $ 1\leq i\leq 3$;} \\
      0\leq v_i-v_1\leq u_i-1, & \hbox{for $ 2\leq i\leq 3$;} \\
      v_2+v_3\leq k-r+v_1,& \hbox{ }\\
      m_{r,j}=\sum_{i=2}^{3} \frac{q^{v_i}-1}{q-1}-\frac{q^{v_1}-1}{q-1}. & \hbox{ }
                                                                       \end{array}
                                                                     \right.
\Biggm{\}}.$$ }

{\bf (e)} 
Assume $1\leq r<u_2.$ Then $d_r(C) = \frac{q^{k}-q^{k-r}}{q-1}-\sum_{i=2}^3\frac{q^{u_i}-q^{u_i-r}}{q-1}.$
By Lemma 18 of \cite{HLL},
we have that  \begin{eqnarray*}
        d_r(C)+\sum_{j=1}^{k-r}\lceil \frac{(q-1)d_r(C)}{q^j(q^r-1)}\rceil &=&  \frac{q^k-1}{q-1}-\sum_{i=2}^3\frac{q^{u_i}-1}{q-1}=n-1. \\
      \end{eqnarray*}

Assume $ r=u_2.$ Then $d_r(C) = \frac{q^{k}-q^{k-r}}{q-1}-\sum_{i=2}^3\frac{q^{u_i}-q^{u_i-r}}{q-1}+1$
and \begin{eqnarray*}
        d_r(C)+\sum_{j=1}^{k-r}\lceil \frac{(q-1)d_r(C)}{q^j(q^r-1)}\rceil &=&  \frac{q^k-1}{q-1}-\sum_{i=2}^3\frac{q^{u_i}-1}{q-1}+1=n. \\
      \end{eqnarray*}
Hence the linear code $C$ is a $u_2$-Griesmer code and an almost Griesmer code.

Suppose there exists an $[n, k,\bar{d}]_q$-linear code $\bar{C}$ such that $\bar{d}\ge d+1.$
Since $d$ is divisible by $q$, we have that $\lceil \frac{\bar{d}}{q}\rceil\ge  \lceil \frac{d+1}{q}\rceil=\lceil\frac{d}{q}\rceil+1.$
Then $$\sum_{i=0}^{k-1}\lceil \frac{\bar{d}}{q^i}\rceil \ge  \sum_{i=0}^{k-1}\lceil \frac{d}{q^i}\rceil+2=n+1,$$
which is a contradiction to the Griesmer bound for the minimum Hamming
distance of linear codes.
Hence the linear code $C$ is a distance-optimal code.

 \end{proof}

\begin{Remark}
Note that, if the set $\mathbf{Z}(q, m_{r,j})$ of Statement \textbf{(c)} in Theorem~\ref{xxd4} is empty, then $A_j^{r}(C) =0$.
When $q\ge 4$,  the size of the set $\mathbf{Z}(q, m_{r,j})$ is less than 2.
\end{Remark}

\begin{Corollary}\label{3}
Assume the notation is as given above.
For an integer $r$ with $1\leq r\leq k$,
then $$ |\{j\,|\, A_j^{r}(C) \neq 0\}| \leq 2\cdot\min\{ u_2u_3,\,(r+1)^2,\,(k-r+1)^2\}.$$ 
\end{Corollary}

\begin{proof}
By Theorem~\ref{xxd4}, we know that
$$ |\{j\,|\, A_j^{r}(C) \neq 0\}| \leq |\{\mathbf{Z}(q, m_{r,j})\,|\, \mathbf{Z}(q, m_{r,j})\neq \emptyset\}|.$$
Then we enumerate all the integers $v_1$, $v_2$ and $v_3$ satisfying $0\leq v_1\leq 1$, 
$$ \max\{0,\,u_i-r\}\leq v_i\leq \min\{u_i,\,k-r\}$$ for $ 2\leq i\leq 3$ and 
      $0\leq v_i-v_1\leq u_i-1$ for $ 2\leq i\leq 3.$
Hence we have that $$ |\{j\,|\, A_j^{r}(C) \neq 0\}| \leq 2\cdot\min\{ u_2u_3,\,(r+1)^2,\,(k-r+1)^2\} .$$
\end{proof}

By Corollary~\ref{3}, we know that the linear codes constructed in Theorem~\ref{xxd4} have few subcode support weights with dimension $r$, when $\min\{ u_2u_3,\,(r+1)^2,\,(k-r+1)^2\}$ is small.

\begin{Example}\label{ee12}
Assume $q=2,$ $s=3$ and $k=5$ in Theorem~\ref{xxd4}.
Let  
{ \tiny $$S_{2,5} =\left(\begin{array}{ccccc ccccc ccccc ccccc ccccc ccccc c}
      1&0&1&0&1&0&1&0&1&0&1&0&1&0&1&0&1&0&1&0&1&0&1&0&1&0&1&0&1&0&1\\
      0&1&1&0&0&1&1&0&0&1&1&0&0&1&1&0&0&1&1&0&0&1&1&0&0&1&1&0&0&1&1 \\
      0&0&0&1&1&1&1&0&0&0&0&1&1&1&1&0&0&0&0&1&1&1&1&0&0&0&0&1&1&1&1\\
      0&0&0&0&0&0&0&1&1&1&1&1&1&1&1&0&0&0&0&0&0&0&0&1&1&1&1&1&1&1&1\\
      0&0&0&0&0&0&0&0&0&0&0&0&0&0&0&1&1&1&1&1&1&1&1&1&1&1&1&1&1&1&1
\end{array}\right)  ,$$

$$G_1=\left(\begin{array}{c}
      1\\
      0 \\
      0\\
      0\\
      0
\end{array}\right) ,\,\,
G_2=\left(\begin{array}{cc cccc}
      1&0&1\\
      0&1&1 \\
      0&0&0\\
      0&0&0\\
      0&0&0
\end{array}\right)\,\, ,\,\,G_3=\left(\begin{array}{cccc ccc}
      1&0&1&0&1&0&1\\
      0&0&0 &0&0&0&0\\
      0&0&0&0&0&0&0\\
      0&1&1&0&1&0&1\\
      0&0&0&1&1&1&1
\end{array}\right)  $$}
{ \tiny  $$ \text{ and } G_4=\left(\begin{array}{ccccc ccccc ccccc}
1&1&1&1&1&1&1&1       &0&0&0&0&0&0&0\\
0&0&0&0&0&0&0&0   &0&0&0&0&0&0&0\\
0&0&0&0&1&1&1&1     &0&0&0&1&1&1&1\\
0&0&1&1&0&0&1&1     &0&1&1&0&0&1&1\\
0&1&0&1&0&1&0&1     &1&0&1&0&1&0&1
\end{array}\right)    .$$ }
Let $C_1$ be the $[22,5,10]_2$-linear code with the generator matrix $[S_{2,5}\backslash [G_2,G_3\backslash G_1]],$ which is on the condition of $u_2=2$ and $u_3=3$ in Theorem~\ref{xxd4}.
 And let $C_2$ be the $[14,5,6]_2$-linear code with the generator matrix $[S_{2,5}\backslash [G_2,G_4\backslash G_1]],$ which is on the condition of $u_2=2$ and $u_3=4$ in Theorem~\ref{xxd4}.
The linear codes $C_1$ and $C_2$ are  $2$-Griesmer codes and  almost Griesmer codes.
Based on the tables in \cite{G}, linear codes $C_1$ and $C_2$ are distance-optimal codes.
By Magma and Theorem~\ref{xxd4}, the subcode support weight distributions of $C_1$ and $C_2$ are  listed in Table~\ref{t2}.

\begin{table}[htbp]
\caption{Parameters of  $C_1$ and $C_2$ in Example~\ref{ee12}}
\label{t2}
\center
{ \small
\begin{tabular}{c |c}

parameters of $C_1$&parameters of $C_2$\\ 
$[22,5,10]_2$&$[14,5,6]_2$ \\ \hline
$1$-SSWD of $C_1$&$1$-SSWD of $C_2$\\
\{[10,6],[11,16],[12,6],[14,2],[16,1]\}&\{[6,6],[7,16],[8,7],[14,1]\}\\ \hline
$2$-SSWD of $C_1$&$2$-SSWD of $C_2$\\
\{[16,60],[17,48],[18,35],[19,8],[20,3],[22,1]\}& \{[10,77],[11,56],[12,7],[14,15]\}\\ \hline
$3$-SSWD of $C_1$&$3$-SSWD of $C_2$\\
 \{[19,48],[20,87],[21,12],[22,8]\}& \{[12,91],[13,28],[14,36]\}\\\hline
$4$-SSWD of $C_1$&$4$-SSWD of $C_2$\\
 \{[21,22],[22,9]\}& \{[13,14],[14,17]\}\\

\end{tabular}
}
\end{table}

\end{Example}

\section{Simplex complement codes of GRS codes}
In this section, we use simplex complement codes of GRS codes to construct a family of $2$-Griesmer codes.
Recall $w(U)$ is the subcode support weight of a subspace $U$.
And the sequence $$[A_1^{r}(C), A_2^{r}(C),\cdots, A_n^{r}(C)]$$ is the {\it $r$-SSWD} of an $[n,k]_q$-linear code $C$.
Assume $S_{q,k}$ is the $k\times \frac{q^k-1}{q-1}$ matrix over $\mathbb{F}_q$ defined as in Section 3.
Let $\gamma_1, \gamma_2,\cdots, \gamma_m$ be $m$ distinct elements of $\mathbb{F}_q$, where $k\leq m\leq q.$ 
Assume  $\alpha_i^T=(\gamma^{k-1}_i, \gamma^{k-2}_i,\cdots ,\gamma_i,1)$ for $1\leq i\leq n$ and 
$$G(\gamma_1, \gamma_2,\cdots, \gamma_m)=[\alpha_1,\alpha_2,\cdots , \alpha_m]$$ is the $k\times m$ matrix defined by $\gamma_1, \gamma_2,\cdots, \gamma_m$.
We can assume that the columns of $G(\gamma_1, \gamma_2,\cdots, \gamma_m)$ are contained in those columns of $S_{q,k}.$


\begin{Theorem}\label{chen}
Assume the notation is as given above and $3\leq k\leq m\leq q$.
Let $\tilde{C}$ be the $[m,k]_q$-linear code generated by $G(\gamma_1, \gamma_2,\cdots, \gamma_m)$, and let $C$ be the $[\frac{q^k-1}{q-1}-m,k]_q$-linear code generated by $S_{q,k}\backslash G(\gamma_1, \gamma_2,\cdots, \gamma_m)$.  Then 
\begin{description}
  \item[(a)]  $d_r(C)=\frac{q^{k}-q^{k-r}}{q-1}-m $ for $1\leq r\leq k.$
  \item[(b)] For  $1\leq r\leq k,$  the $r$-SSWD  of $C$ satisfies
  $$A_j^{r}(C)=\left\{
                     \begin{array}{ll}
                       0, & \hbox{$0\leq j<  \frac{q^{k}-q^r}{q-1}-m$ ;} \\
                       A_{\frac{q^{k}-q^r}{q-1}-j}^{r}(\tilde{C}), & \hbox{$\frac{q^{k}-q^r}{q-1}-m\leq j\leq\frac{q^{k}-q^r}{q-1}-m+k-r  $;}\\
0, & \hbox{$  \frac{q^{k}-q^r}{q-1}-m+k-1   <j\leq \frac{q^{k}-1}{q-1}$ .} 
                     \end{array}
                   \right. $$
  \item[(c)]  If $m<q$, then $C$ is  a Griesmer code. If $m=q$, then $C$ is a $2$-Griesmer code and an almost Griesmer code.
  \item[(d)] The weight distribution  of $C$ is 
$$A_j(C)=\binom{m}{q^{k-1}-j}\sum_{t=0}^{q^{k-1}+k-j-m-1}(-1)^t\binom{q^{k-1}-j}{t}(q^{q^{k-1}+k-j-m-t}-1)$$
 for $q^{k-1}-m\leq j\leq q^{k-1}+k-m-1,$ otherwise $A_j(C)=0$.
\end{description}

\end{Theorem}
\begin{proof}
{\bf (a)-\bf (b)}
Note that $\tilde{C}$ is an MDS code.
Let $\mathfrak{S}_{q,k}$ be the $[\frac{q^k-1}{q-1},k]_q$-linear code generated by $S_{q,k}$.
And we know that the support weight of any subspace with dimension $r$ of $\mathfrak{S}_{q,k}$ is $\frac{q^{k}-q^{k-r}}{q-1}$ and $$d_r(\mathfrak{S}_{q,k})=\frac{q^{k}-q^{k-r}}{q-1}$$ for $1\leq r\leq k.$

By the definition of  $G(\gamma_1, \gamma_2,\cdots, \gamma_m)$ and $S_{q,k}\backslash G(\gamma_1, \gamma_2,\cdots, \gamma_m)$,
there exists  a linear isomorphism $\psi$ from $C$ to $\tilde{C}$ such that $[\psi(\mathbf{c}),\mathbf{c}]\in \mathfrak{S}_{q,k}$ for any $\mathbf{c}\in C.$
If $\mathbf{c}\neq \mathbf{0}$, then $\psi(\mathbf{c})\neq \mathbf{0}$ and $w(\mathbf{c})=q^{k-1}-w(\psi(\mathbf{c}))$. 
Hence  \begin{eqnarray*}
  d_1(C) &=& q^{k-1}-\max\{w(\psi(\mathbf{c}))\,|\, \mathbf{0} \neq \mathbf{c}\in C\} \\
   &=&q^{k-1}-\max\{w(\tilde{\mathbf{c}})\,|\,\mathbf{0} \neq\tilde{\mathbf{c}}\in \tilde{C} \}\\
   &=&q^{k-1}-m .
\end{eqnarray*}

For any subspace $U\in {\rm SUB}^{r}(C)$, we know that 
$$w(U)=\frac{q^{k}-q^{k-r}}{q-1} - w(\psi(U))$$
and \begin{eqnarray*}
  d_r(C) &=& \frac{q^{k}-q^{k-r}}{q-1}-\max\{w(\psi(U))\,|\, U\in {\rm SUB}^{r}(C)\} \\
   &=&\frac{q^{k}-q^{k-r}}{q-1}-\max\{w(V)\,|\,V\in {\rm SUB}^{r}(\tilde{C})\}\\
   &=&\frac{q^{k}-q^{k-r}}{q-1}-m .
\end{eqnarray*}

Hence, for  $1\leq r\leq k,$  $$A_j^{r}(C)=\left\{
                     \begin{array}{ll}
                       0, & \hbox{$0\leq j<  \frac{q^{k}-q^r}{q-1}-m$ ;} \\
                       A_{\frac{q^{k}-q^r}{q-1}-j}^{r}(\tilde{C}), & \hbox{$\frac{q^{k}-q^r}{q-1}-m\leq j\leq\frac{q^{k}-q^r}{q-1}-m+k-r  $;}\\
0, & \hbox{$  \frac{q^{k}-q^r}{q-1}-m+k-1   <j\leq \frac{q^{k}-1}{q-1}$ .} 
                     \end{array}
                   \right. $$

{\bf (c)} 
Since $m\leq q,$ we know that
\begin{eqnarray*}
  d_2(C)+\sum_{i=1}^{k-2}\lceil \frac{(q-1)d_2(C)}{q^i(q^2-1)}\rceil&=&  \frac{q^{k}-q^{k-2}}{q-1}-m+\sum_{i=1}^{k-2}\lceil \frac{q^{k}-q^{k-2}-(q-1)m}{q^i(q^2-1)}\rceil\\
   &=&  \frac{q^{k}-q^{k-2}}{q-1}-m+\sum_{i=1}^{k-2}\lceil q^{k-2-i}-\frac{m}{q^i(q+1)}\rceil \\
    &=&  \frac{q^{k}-q^{k-2}}{q-1}-m+\sum_{i=1}^{k-2}q^{k-2-i} \\
     &=&  \frac{q^{k}-1}{q-1}-m ,
\end{eqnarray*}
which is the length of $C$.

Note that
\begin{eqnarray*}
  \sum_{i=0}^{k-1}\lceil \frac{d_1(C)}{q^i}\rceil&=&  q^{k-1}-m+\sum_{i=1}^{k-1}\lceil q^{k-1-i} -\frac{m}{q^i}\rceil.
\end{eqnarray*}
If $m<q$, then $$\sum_{i=0}^{k-1}\lceil \frac{d_1(C)}{q^i}\rceil=\frac{q^{k}-1}{q-1}-m,$$ which is the length of $C$.
Then $C$ is  a Griesmer code.

If $m=q$, then $$\sum_{i=0}^{k-1}\lceil \frac{d_1(C)}{q^i}\rceil=\frac{q^{k}-1}{q-1}-m-1,$$ which is less than the length of $C$.
Then $C$ is  a $2$-Griesmer code and an almost Griesmer code.

{\bf (d)} Note that $\tilde{C}$ is an MDS code. The weight distribution of MDS codes is provide in Theorem 7.4.1 of \cite{HP}.
Then $A_j(C)=0$ for $1\leq j\leq q^{k-1}-m-1$ or $q^{k-1}+k-m\leq j\leq \frac{q^{k}-1}{q-1}-m.$
And
$$A_j(C)=A_{q^{k-1}-j}(\tilde{C})
=\binom{m}{q^{k-1}-j}\sum_{t=0}^{q^{k-1}+k-j-m-1}(-1)^t\binom{q^{k-1}-j}{t}(q^{q^{k-1}+k-j-m-t}-1)$$
 for $q^{k-1}-m\leq j\leq q^{k-1}+k-m-1.$

\end{proof}

In the following corollary, we determine some subcode support weight distributions of simplex complement codes of GRS codes.
\begin{Corollary}\label{rr3}
Assume the notation is as given above and $3\leq k\leq m\leq q$.
Let $\tilde{C}$ be the $[m,k]_q$-linear code generated by $G(\gamma_1, \gamma_2,\cdots, \gamma_m)$, and let $C$ be the $[n,k]_q$-linear code generated by $S_{q,k}\backslash G(\gamma_1, \gamma_2,\cdots, \gamma_m)$, where $n=\frac{q^k-1}{q-1}-m$. 
\begin{description}
  \item[(a)]  If $r=k-1$, then 
$A_j^{k-1}(C)=\left\{
                              \begin{array}{ll}
                                m, & \hbox{$j=n$;} \\
                                \frac{q^k-1}{q-1}-m, & \hbox{$j=n-1$;} \\
                                0, & \hbox{  $1\leq j\leq n-1$.}
                              \end{array}
                            \right.$ 
  \item[(b)] If $r=k-2$, then $$A_j^{k-2}(C)=\left\{
                              \begin{array}{ll}
                                \frac{(q^k-1)(q^k-q)}{(q^2-1)(q^2-q)}-m\frac{q^{k-1}-1}{q-1}+\binom{m}{2}, & \hbox{$j=q^{k-2}(q+1)-m$;} \\
m(\frac{q^{k-1}-1}{q-1}-m+1), & \hbox{$j=q^{k-2}(q+1)-m+1$;} \\
\binom{m}{2}, & \hbox{$j=q^{k-2}(q+1)-m+2$;} \\
                                0, & \hbox{  otherwise.}
                              \end{array}
                            \right.$$
\end{description}
\end{Corollary}
\begin{proof}
Statements {\bf(a)} is a direct result of Lemma~\ref{weight}. 

{\bf(b)}
Let $\tilde{C}$ be the $[m,k]_q$-linear code generated by $G(\gamma_1, \gamma_2,\cdots, \gamma_m).$
By Lemma~\ref{weight}, we know that $$A_j^{k-2}(\tilde{C})=\left\{
                              \begin{array}{ll}
                                   \frac{(q^k-1)(q^k-q)}{(q^2-1)(q^2-q)}-m\frac{q^{k-1}-1}{q-1}+\binom{m}{2}, & \hbox{$j=m$;} \\
m(\frac{q^{k-1}-1}{q-1}-m+1), & \hbox{$j=m-1$;} \\
\binom{m}{2}, & \hbox{$j=m-2$;} \\
                                0, & \hbox{  $1\leq j\leq m-3$.}
                              \end{array}
                            \right.$$
By Statement {\bf(c)} of Theorem~\ref{chen}, we have that 
$$A_j^{k-2}(C)=A_{q^{k-2}(q-1)-j}^{k-2}(\tilde{C})=\left\{
                              \begin{array}{ll}
                                \frac{(q^k-1)(q^k-q)}{(q^2-1)(q^2-q)}-m\frac{q^{k-1}-1}{q-1}+\binom{m}{2}, & \hbox{$j=q^{k-2}(q-1)-m$;} \\
m(\frac{q^{k-1}-1}{q-1}-m+1), & \hbox{$j=q^{k-2}(q-1)-m+1$;} \\
\binom{m}{2}, & \hbox{$j=q^{k-2}(q-1)-m+2$;} \\
                                0, & \hbox{  otherwise.}
                              \end{array}
                            \right.$$

\end{proof}

\begin{Example}\label{ee132}
Let $C_1$ be the $[10,3,6]_3$-linear code with the generator matrix 
$$S_{3,3}\backslash G(0, 1,2)=\left(\begin{array}{ccccc ccccc }
1&0&1&2& 1&2& 0&2 & 0&2\\
0&1&1&1& 0&0& 1&1 & 2&2\\
0&0&0&0&1&1& 1&1  & 1&1
\end{array}\right) .$$
Let $C_2$ be the $[151,4,120]_5$-linear code with the generator matrix $S_{5,4}\backslash G(0, 1,2,3,4).$
Then $C_1$ and $C_2$ are $2$-Griesmer codes and almost Griesmer codes.
By Magma and Corollary~\ref{rr3}, the subcode support weight distributions of $C_1$ and $C_2$ are listed in Table~\ref{t6}.

\begin{table}[htbp]
\caption{Parameters of  $C_1$ and $C_2$ in Example~\ref{ee132}}
\label{t6}
\center

{\small
\begin{tabular}{c|c }
parameters of $C_1$&parameters of $C_2$\\
$[10,3,6]_3$  &$[151,4,120]_5$\\\hline
$1$-SSWD of $C_1$&$1$-SSWD of $C_2$\\
\{[6,4],[7,6],[8,3]\}&\{[120,51],[121,65],[122,30],[123,10]\}\\ \hline
$2$-SSWD of $C_1$&$2$-SSWD of $C_2$\\
\{[9,10],[10,3]\}&\{[145,661],[146,135],[147,10]\}\\\hline
$3$-SSWD of $C_1$&  $3$-SSWD of $C_2$\\
   \{[10,1]\}& \{[150,151],[151,5]\}\\
\end{tabular}
}
\end{table}
\end{Example}

\section{Conclusions and remarks}
In this work, we construct four families of distance-optimal few-weight linear codes  and their subcode support weight distributions are determined as follows.
\begin{itemize}
  \item The Griesmer codes constructed in \textbf{Theorem~\ref{xxd2}} have the parameter $$ [t\frac{q^{k}-1}{q-1}-\sum_{i=1}^s\frac{q^{u_i}-1}{q-1},\,k,\,tq^{k-1}-\sum_{i=1}^sq^{u_i-1}]_q,$$ 
      where $1\leq s\leq t$ and $0< u_1< u_2< \cdots< u_s<k.$ 
      Assume $s=1$ and $C$ is a Griesmer code constructed in \textbf{Theorem~\ref{xxd2}},
we have that 
$$
    |\{j\,|\, A_j^{r}(C) \neq 0\}| \leq \min\{ u_1+1,\,r+1,\,k-r+1,\,k-u_1+1\} 
$$
    for  $1\leq r\leq k$.
  \item For $r\ge 2,$ the $r$-Griesmer codes constructed in \textbf{Theorem~\ref{xxd}} have the parameter $$[t\frac{q^k-1}{q-1}-\sum_{i=2}^s\frac{q^{u_i}-1}{q-1}+1,\,k,\,tq^{k-1}-\sum_{i=2}^sq^{u_i-1}]_q,$$
       where $2\leq s\leq t+1$ and $1=u_1< u_2< \cdots< u_s< k.$  
      Assume $s=2$ and $C$ is a $r$-Griesmer code constructed in \textbf{Theorem~\ref{xxd}},
we have that $$ |\{j\,|\, A_j^{r}(C) \neq 0\}| \leq \min\{ 2u_2,\,2(r+1),\,2(k-r+1),\,2(k-u_2+1)\}$$ 
for  $1\leq r\leq k$.
      
   \item For $r\ge 2,$ the $r$-Griesmer codes constructed in \textbf{Theorem~\ref{xxd4}} have the parameter $$[\frac{q^k-1}{q-1}-\frac{q^{u_2}-1}{q-1}-\frac{q^{u_3}-1}{q-1}+1\,,k,\,q^{k-1}-q^{u_2-1}-q^{u_3-1}]_q,$$ 
       where  $ 1=u_1< u_2<  u_3<k$  such that $u_2+u_3\leq k+1.$
       Assume  $C$ is a $r$-Griesmer code constructed in \textbf{Theorem~\ref{xxd4}},
we have that $$ |\{j\,|\, A_j^{r}(C) \neq 0\}| \leq 2\cdot\min\{ u_2u_3,\,(r+1)^2,\,(k-r+1)^2\}$$
for  $1\leq r\leq k$.
   The minimum Hamming distance of the dual code of a $r$-Griesmer code constructed in \textbf{Theorem~\ref{xxd4}}
   is large than the minimum Hamming distance of the dual code of a $r$-Griesmer code constructed in \textbf{Theorem~\ref{xxd}}.
   
   \item For $1\leq r\leq  2,$ the $r$-Griesmer codes constructed in \textbf{Theorem~\ref{chen}} have the parameter $$[\frac{q^k-1}{q-1}-m,\,k, \,q^{k-1}-m ]_q,$$ 
       where $3\leq k\leq m\leq q$.
\end{itemize}

 The subcode support weight distributions of the linear codes constructed in \textbf{Theorem~\ref{chen}} are not totally determined. It is interesting to determine the subcode support weight distributions of the linear codes constructed in \textbf{Theorem~\ref{chen}} completely.
 The method of Solomon-Stiffler \cite{SS} seem to be a promise way to construct distance-optimal few-weight linear codes  such that  their subcode support weight distributions can be determined. A particularly interesting problem is to construct more such codes by the method of Solomon-Stiffler.

\vskip 4mm

\noindent {\bf Acknowledgement.} This work was supported by The fellowship of China National
Postdoctoral Program for Innovative Talents (BX20240142), The Guangdong Key Laboratory of Data Security and Privacy Preserving (Grant No. 2023B1212060036)
and The National Natural Science Foundation of China (Grant No. 62032009, Grant No. 12271199, Grant No. 12171191, Grant No. 61902149 and Grant No. 62311530098).


\end{document}